\documentclass[a4paper,12pt, reqno]{amsart}
\usepackage{mathrsfs}
\usepackage{amsmath}
\usepackage[height=22cm, width=16.5cm, hmarginratio={1:1}]{geometry}
\usepackage{diagbox}
\usepackage{tikz}
\usetikzlibrary{positioning, shapes.geometric}
\usepackage{newclude}
\usepackage{appendix}
\usepackage{extarrows}

\usepackage[hidelinks,hyperindex]{hyperref}
\usepackage{array}
\usepackage{dsfont}
\newtheorem{manualtheoreminner}{Theorem}
\newenvironment{manualtheorem}[1]{%
	\renewcommand\themanualtheoreminner{#1}%
	\manualtheoreminner
}{\endmanualtheoreminner}

\newcommand{\<}{\langle}
\renewcommand{\>}{\rangle}

\def\FV{\mathcal{V}}
\def\rR{\mathscr R}
\def\rV{\mathscr V}

\def\Gr{\mathrm{Gr}}
\def\KW{\mathrm{KW}}

\newcolumntype{C}[1]{>{\centering\arraybackslash$}m{#1}<{$}}
\newlength{\mycolwd}                                         
\newlength{\mycolwdm}                                         
\newlength{\mycolwda}                                         
\newlength{\mycolwdb}                                         
\newlength{\mycolwdc}                                         
\newlength{\mycolwdd}                                         
\newlength{\mycolwddd}                                         
\newlength{\mycolwddc}                                         
\newcommand{\givz}{{u}}

\newcommand{\Tr}{\mathrm{Tr}}

\newcommand{\Wit}{\mathscr{W}}

\usepackage {color,graphicx,psfrag, verbatim,amssymb,amscd,enumerate,subfigure}

\def\beq{\begin{equation}}
\def\eeq{\end{equation}}
\newcommand{\be}{\begin{equation*}}
\newcommand{\ee}{\end{equation*}}

\DeclareMathOperator{\Res}{Res}

\DeclareMathOperator{\diag}{diag}

\newtheorem{dummy}{}[section]
\newtheorem{lemma}[dummy]{Lemma}
\newtheorem{proposition}[dummy]{Proposition}
\newtheorem{theorem}[dummy]{Theorem}
\newtheorem{corollary}[dummy]{Corollary}

\theoremstyle{definition}

\newtheorem{definition}[dummy]{Definition}
\newtheorem{example}[dummy]{Example}

\newtheorem{remark}[dummy]{Remark}

\usepackage{graphicx}

\newcommand{\pd}{\partial}

\newcommand{\cD}{\mathcal{D}}

  \newcommand{\cC}{\mathcal{C}}

\newcommand{\Mbar}{\overline{\mathcal M}}

\author{Shuai Guo}
\address{Shuai Guo, School of Mathematics Science, Peking University,
	No 5. Yiheyuan Road,  Beijing 100871, China}
\email{guoshuai@math.pku.edu.cn}

\author{Ce Ji}
\address{Ce Ji, Department of Mathematical Sciences, Tsinghua University,
	No 1. Tsinghuayuan, Beijing 100084, China}
\email{cji@tsinghua.edu.cn}

\author{Chenglang Yang}
\address{Chenglang Yang, Institute for Math and AI, Wuhan University}
\email{yangcl@whu.edu.cn}

\author{Qingsheng Zhang}
\address{Qingsheng Zhang, School of Sciences,  Great Bay University,
	Great Bay Institute for \newline Advanced Study,
	Dongguan, 523000, China}
\email{zhangqingsheng@gbu.edu.cn}

\begin{document}

\title{Generalized Kontsevich model, topological recursion, and $r$-spin theory}

\begin{abstract}
	By employing polynomial-reduced KP integrability, combined with the string equation, this work establishes explicit relationships between the generalized Kontsevich model, the topological recursion of the spectral curve, and the geometry of moduli spaces of $r$-spin curves. 
	For the generalized Kontsevich model with a polynomial potential, we derive an explicit formulation and provide a proof of these widely expected correspondences.
	Furthermore, the method is extended to the cases with admissible deformed potentials, where the corresponding geometric theory is a deformed version of $r$-spin theory.
\end{abstract}

\maketitle

\setcounter{section}{-1}
\setcounter{tocdepth}{1}

\section{Introduction}
The integrability is one of the most important structures in modern mathematical physics and has led to great progress in many related areas, particularly in the geometry of moduli spaces.
The famous Witten Conjecture~\cite{Wit91} states that,
the generating series of certain intersection numbers over the moduli spaces of stable curves satisfies the KdV hierarchy.
This conjecture was first proved by Kontsevich~\cite{Kont92}.
Later, the explicit relationships between enumerative geometries and integrable hierarchy were intensively studied in the literature.
For example, the generalized Witten Conjecture for $r$-spin theory \cite{FJR13,FSZ06,Wit93}, the Toda hierarchy conjecture of the Gromov--Witten theory of $\mathbb {CP}^1$~\cite{DZ04,EHY95,Get04,OP06}, 
and the integrable structure of the open Gromov--Witten theory of toric Calabi--Yau 3-folds~\cite{ADKMV,DZ,WYZ}.
Recently,
a generalization of the Witten Conjecture through spectral curves has been proposed in \cite{GJZ23}, 
which establishes the polynomial-reduced KP integrability of certain enumerative geometry via Eynard--Orantin's topological recursion~\cite{EO07}.

In the proof of the Witten Conjecture~\cite{Kont92},
Kontsevich introduced a new matrix model, now known as the Kontsevich model.
This model was then generalized to the Generalized Kontsevich Model (GKM)~\cite{KMMMZ,Kont92} (see also~\cite{ACM12,Ale21a}), 
which is defined as a Hermitian matrix model with an external field and a potential $V(z)$.
When the potential is specified as $V(z)=\frac{z^{r+1}}{r(r+1)}$, 
the GKM coincides with the generating series of Witten's $r$-spin theory, 
a result that follows directly from the generalized Witten Conjecture. 
It further follows from~\cite{BE13,BE17} that, with this potential, the GKM coincides with the Bouchard--Eynard topological recursion of the $r$-spin curve. 
In particular, for the Kontsevich model case where $V(z)=\frac{z^3}{6}$, its correspondence with the topological recursion of the Airy curve is well-known ~\cite{EO07} (see also \cite{Zhou13}).

While the integrability of the GKM is preserved under deformations of the potential $V(z)$, the corresponding enumerative geometry and topological recursion are generally not fully understood, even in the case of polynomial potentials.
In this work, we exploit the polynomial-reduced KP integrability, together with the string equation, to establish an explicit relationship between the GKM, the $r$-spin theory, and the topological recursion.
We will show that this method is applicable not only to the polynomial potential cases, but also to the deformed cases introduced by Alexandrov~\cite{Ale21a}.

\subsection{Main results}
Given a potential function $V(z)$, the associated GKM is defined by
$$
Z_{V,N}(\Lambda)=\frac{1}{\mathcal{C}_V}
\int [d\Phi] \exp\Big(-\frac{1}{\hbar}
\Tr\big(V(\Phi)-\Phi V'(\Lambda)\big)\Big),
$$
where the integral is over the space of $N\times N$ Hermitian matrices, $[d\Phi]$ is the Lebesgue measure,
and $\Lambda=\diag\{\lambda_1,\cdots,\lambda_N\}$.
See \S \ref{sec:GKMM} for more details of the definition of this matrix model.
Introduce the Miwa variables $T_k=\frac{1}{k}\Tr(\Lambda^{-k})$, $k\geq1$.
By letting $N\rightarrow+\infty$, $T_k$ can be viewed as independent variables.
The partition function of the GKM is defined by
$$
Z_{V}({\bf T};\hbar):= Z_{V,N}(\Lambda)|_{N\to +\infty}.
$$ 
Furthermore, given the partition function, the associated connected genus $g$ $n$-differential $\omega^{V}_{g,n}(z_1,\cdots,z_n)$ for the GKM is defined by:
$$
\omega^{V}_{g,n}(z_1,...,z_n):=
[\hbar^{2g-2+n}]\sum_{k_1,...,k_n=1}^{\infty}
\left.\frac{\partial^n \log Z_V(\mathbf{T};\hbar)}
{\partial T_{k_1} \dots \partial T_{k_n}}
\right|_{\mathbf{p}=0} \frac{dz_1^{-k_1}}{k_1} \cdots \frac{dz_n^{-k_n}}{k_n}.
$$
Here $[h^k]f(\hbar)$ stands for the coefficient of $\hbar^k$ in the series $f(\hbar)$.

The first main result of the present paper is
\begin{manualtheorem}{1}[=Theorem~\ref{thm:GKM-TR}+Theorem~\ref{thm:GKM-TR-deformed}]
\label{thm:main1}
	For a generic polynomial potential  $V(z)\in \mathbb C[z]$ of degree $r+1\geq 3$, we denote by
	$\omega^{V}_{g,n}$ the multi-differentials of the GKM with potential $V(z)$ and by $\omega_{g,n}$ the multi-differentials of the topological recursion of the following spectral curve (see~\S \ref{sec:TR-CohFT-KP})
	\begin{align}\label{eqn:spec data}
		\cC=(\mathbb{P}^1,\quad  x(z)=V'(z),\quad y(z)=z).
	\end{align}
	Then for any $2g+n-2>0$ and $z_1,\cdots,z_n$ near $\infty$, we have
	\beq\label{eqn:omegaV=omega}
	\omega^{V}_{g,n}(z_1,...,z_n)=\omega_{g,n}(z_1,...,z_n).
	\eeq
	Furthermore, for an admissible deformed potential $V_r^{\epsilon}(z)$ such that its second derivative is a rational function of the form in \eqref{eqn:def Uepsilon},
	the identity ~\eqref{eqn:omegaV=omega} still holds for $n\geq 1$.
\end{manualtheorem}

In the above theorem,
we provide a direct and explicit formulation of the connection between the GKM with polynomial potential and the topological recursion of the spectral curve \eqref{eqn:spec data}.
Although the connection for the polynomial cases is widely expected, to the best of our knowledge, such a clear result and fully detailed proof does not exist in the literature.
In \cite{EO07},
a relation between matrix models with external field and topological recursion has been studied (see their \S 10.4).
Moreover, a similar spectral curve whose $x(z)$ is also a polynomial but with a distinct function $y(z)$ was studied in \cite{BCEG25} (see their Definition 3.2), and the authors established its relation to the generalized Kontsevich graphs.
It is worth mentioning that our result also holds for the deformed potential case (see \S \ref{sec:deformed-GKM} for more details),
which is proposed by Alexandrov in \cite{Ale21a} when he studied the relation between GKM and Hodge integrals.

As a corollary of Theorem~\ref{thm:main1},
we are able to study Gukov--Su{\l}kowski's proposal about a kind of quantization of the spectral curve \cite{GS12}.
We refer the reader to~\cite{EGMO24} and the references therein for related results and further discussions on this topic.

\begin{corollary}
	For the spectral curve $\cC=(\mathbb{P}^1, x(z), y(z)=z)$,
	where $x(z)$ is an arbitrary polynomial,
	the Gukov--Su{\l}kowski conjecture holds.
\end{corollary}

Now we turn to the geometric side: the Witten's $r$-spin theory, which is the intersection theory on the moduli space of the $r$-spin curves.
This theory was introduced by Witten~\cite{Wit93} and its mathematical definition was established by many mathematicians~\cite{Jar00,PV01,Chi06,FJR13}.
The generalized Witten Conjecture states that via an explicit and simple coordinate transformation ${\bf t=t(T)}$, 
the total descendent potential $\cD^{r}({\bf t};\hbar)$ of the $r$-spin theory gives a tau-function of the $r$KdV hierarchy which is uniquely determined by the string equation. 
This conjecture was proved by Faber--Shadrin--Zvonkine~\cite{FSZ06} (see also~\cite{FJR13}).

In this paper, we consider the $t$-shifted descendent invariants $\<-\>_{g,n}^{r,t}$ of the $r$-spin theory 
(see \S \ref{sec:rspin-geo} for the definition). 
Then we establish the connection between the shifted $r$-spin theory and the GKM with polynomial potential.

The second main result in present paper is
\begin{manualtheorem}{2}\label{thm:main2}
	For a polynomial $V(z)\in \mathbb C[z]$ of degree $r+1\geq 3$, we denote by
	$\omega^{V}_{g,n}$ the multi-differentials of the GKM with potential $V(z)$.
	Consider paths $\tilde\gamma_i=e^{\frac{2i\pi \sqrt{-1}}{r}}[0,\infty)\subset \mathbb C$, and $\gamma_i:=\tilde\gamma_0-\tilde\gamma_i$, we introduce a sequence of classes by
	$$
	\Phi_i(\givz):=\frac{1}{\sqrt{-r}}\sum_{j=1}^{r-1}\big(e^{\frac{2ij\pi\sqrt{-1}}{r}}-1\big)\, \Gamma\Big(\frac{j}{r}\Big)\, \givz^{\frac{j}{r}}\, \phi_{r-1-j},
	$$
	then we have for $2g-2+n>0$,
	$$
	\int_{\gamma_{i_1}}\cdots \int_{\gamma_{i_n}}e^{-\sum_i V'(z_{i})/\givz_i}
	\omega_{g,n}^{V}(z_1,\cdots,z_n)
	=C_r^{2g-2+n}\cdot \left\<\frac{\Phi_{i_1}(\givz)}{\givz_1-\psi_1},\cdots,\frac{\Phi_{i_n}(\givz)}{\givz_n-\psi_n}\right\>_{g,n}^{r,t(a)},
	$$
	where $C_r$ is a constant determined by $V(z)$, $t(a)$ denotes the mirror map identifying the parameters of the $r$-spin theory with those of $V(z)$, and $\<-\>^{r,t}_{g,n}$ stands for the $t$-shifted descendent invariants of the $r$-spin theory.
	See \S \ref{sec:Geo-GKM} for explicit explanation of notations.
\end{manualtheorem}

We also consider the GKM with the deformed potential introduced by Alexandrov~\cite{Ale21a}. For an admissible deformed potential $V^{\epsilon}_{r}(z)$ (see \S \ref{sec:deformed-GKM} for details), one can introduce an $R$-matrix $\mathscr R^{r,\epsilon}$ and a $T$-vector $\mathscr T^{r,\epsilon}$, both of which are determined by $V_r^{\epsilon}$. The action of these two objects on Witten's $r$-spin class gives rise to a new cohomological field theory (CohFT), whose partition function we denote by $\mathscr D^{r,\epsilon}({\bf t};\hbar)$.

The third main result of this paper is
\begin{manualtheorem}{3}\label{thm:main3}
The potential function of the GKM with the admissible deformed potential $V_r^{\epsilon}(z)$ coincides with the one defined by the above quantization action on the $r$-spin theory:
	$$
	Z_{V_r^{\epsilon}}\big({\bf T};\hbar/\sqrt{-r}\, \big)
	=\mathscr D^{r,\epsilon}(\hbar\cdot {\bf t(T)};\hbar)/\mathscr D^{r,\epsilon}({\bf 0};\hbar).
	$$
	See \S \ref{sec:deformed-GKM-geo} for explicit explanation of notations.
\end{manualtheorem}

\subsection{Outline of the paper}
This paper is organized as follows.  
In \S~\ref{sec:KP}, we review the KP hierarchy and Sato Grassmannian.
In \S~\ref{sec:GKMM}, we recall the GKM and its integrability, after which we derive a differential version of the string equation for the model.
Sections~\ref{sec:TR-GKM} and \ref{sec:rspin-GKM} investigate the connections between the GKM with polynomial potential and, respectively, topological recursion and $r$-spin theory.
Finally, in \S~\ref{sec:deformed-GKM}, we generalize these results to the deformed cases.

\textbf{Acknowledgements.}
This work is partially supported by the National Key Research and Development Program of China 2023YFA1009802 
and the National Natural Science Foundation of China 12225101, 12401079.

\section{Review of the KP hierarchy and Sato Grassmannian} 
\label{sec:KP}
In this section, we recall some basic properties of the KP hierarchy,
especially its relations to the Sato Grassmannian and Kac--Schwarz operators.
We recommend the references \cite{DJM00,KS91,Sato81}.

\subsection{Fermionic Fock space and Sato Grassmannian}
\label{subsec:fermionic tau}
In this subsection,
we review the fermionic Fock space and Sato Grassmannian.
The fermionic Fock space is an infinite dimensional linear space.
It is known that, the standard Pl{\"u}cker embedding gives an embedding from the ordinary Grassmannian $G(k,n)$
to the projective space $\mathbb{P}(\wedge^k \mathbb{C}^n)$.
In the Sato's theory \cite{Sato81}, 
Sato Grassmannian comes from an infinite dimensional version of this embedding,
which provides a beautiful description of the tau-functions of the KP hierarchy.

The fermionic Fock space is an infinite dimensional generalization of the linear space $\wedge^k \mathbb{C}^n$ spanned by semi-infinitely many wedges.
This name comes from an action of fermionic operators on this space.
See \S 3 in \cite{DJM00} for more details.
We start from an infinite dimensional vector space
\begin{align*}
	\FV=z^{1/2}\cdot \mathbb{C}[\![z^{-1}]\!][z].
\end{align*}
It has a basis $\{z^{k+1/2}=\underline{k+1/2}\  |{k\in\mathbb{Z}}\}$.
A sequence of decreasing half integers $S=\{s_1>s_2>\cdots\}$ is called admissible if
$$
|S_+| = | \{s_i\in S, s_i>0\}|<\infty,\quad  |S_-|=  | \{s_i\in \mathbb{Z}_{\leq 0}-\tfrac{1}{2}, s_i\notin S\}|<\infty.
$$
Then, the fermionic Fock space $\mathcal{F}$ is the space of the admissible semi-infinite dimensional subspace of $\FV$.
That is to say,
\begin{align*}
	\mathcal{F}:=\text{span}_{\mathbb{C}}
	\{|S\rangle=\underline{s_1}\wedge \underline{s_2}\wedge\underline{s_3}\wedge\cdots, S\ \text{admissible}\},
\end{align*}
where we admit the number of summands to be infinity.
For an element $|S\rangle\in\mathcal{F}$,
its charge is defined by
\begin{align*}
	\text{charge}(|S\rangle):=|S_+|-|S_-|.
\end{align*}
We denote by $\mathcal{F}^{(0)}$ the charge 0 subspace of $\mathcal{F}$.
In this paper, we restrict us to this subspace for convenience.

Following Sato's theory \cite{Sato81},
the set of tau-functions of the KP hierarchy can be described by the Sato Grassmannian,
which is embedded in $\mathbb{P}\mathcal{F}^{(0)}$.
Below,
we use the Sato Grassmannian to describe tau-functions of the KP hierarchy.

An element $|v\rangle\in\mathcal{F}^{(0)}$ is call a fermionic tau-function of the KP hierarchy if and only if $|v\rangle$ is of the following form
\begin{align}\label{eqn:fi wedge}
	f_0(z)\wedge f_1(z) \wedge \cdots,
\end{align}
where each $f_i(z)\in \FV=z^{1/2}\cdot\mathbb{C}[\![z^{-1}]\!][z]$.
Intuitively,
up to a non-zero constant,
the vector $|v\rangle$ is in the image of the infinite dimensional version of Pl{\"u}cker embedding.
That means the vector $|v\rangle\in\mathcal{F}^{(0)}$ (up to a non-zero constant) can be equivalently described as a semi-infinite dimensional subspace of $\FV$.
We call these functions $f_i(z)$ the basis vectors for the fermionic tau-function $|v\rangle$.
It should be noticed that these basis vectors are not unique.

In this paper,
we focus on the fermionic tau-functions in the big cell of the Sato's Grassmannian $\Gr^{(0)}$.
Elements in the big cell means it has an admissible basis description.
That is to say,
the fermionic tau-function $|v\rangle$ is of the form
\begin{align}\label{eqn:v as phi}
	|v\rangle=c\cdot\varphi_0(z)\wedge \varphi_1(z) \wedge \cdots
\end{align}
such that
\begin{equation}\label{adba}
	\varphi_k = z^{k+\frac{1}{2}}+\sum_{i=-k+1}^{\infty}a_{k,i}z^{-i+\frac{1}{2}},
	\quad a_{k,i}\in\mathbb{C},
	\quad k=0,1,2,\cdots.
\end{equation}
In this case,
the set $\{\varphi_k(z)\}_{k=0}^\infty$ is called an admissible basis of $|v\rangle$.

Notice that,
the admissible basis for a fermionic tau-function is still not unique.
Actually,
$\varphi_i$ could be replaced by $\varphi_i$ adding any linear summation of $\varphi_j$ with $0\leq j<i$.
To get a unique choice of the admissible basis,
we introduce the so-called canonical basis $\{\tilde\varphi_k\}_{k=0}^{\infty}$,
which is an admissible basis and additionally satisfies
\begin{equation}\label{canobasis}
	\tilde\varphi_k = z^{k+\frac{1}{2}}+\sum_{i=0}^{\infty}b_{k,i}z^{-i-\frac{1}{2}}, \quad k=0,1,2,\cdots.
\end{equation}
These coefficients $b_{k,i}$ are referred to as the affine coordinates of the fermionic tau-function $|v\rangle$. 
Equation \eqref{eqn:v as phi} implies that the fermionic tau-function $|v\rangle$ is uniquely determined by the coordinates $b_{k,i}$ and the constant $c$. 
Specifically, the basis of the form \eqref{canobasis} induces a canonical isomorphism from the big cell of the Sato Grassmannian $\Gr^{(0)}$ to an affine space, 
under which the numbers $b_{k,i}$ serve as coordinates in the affine space.

\subsection{Boson-fermion correspondence}
The boson-fermion correspondence provides an isomorphism of the (charge zero) fermionic Fock space and the space of formal power series of variables $\mathbf{T}=(T_1,T_2,...)$.
More precisely,
there is an isomorphism
\begin{align*}
	\Phi:\quad \mathcal{F}^{(0)} \quad&\rightarrow \quad\mathbb{C}[\![\mathbf{T}]\!]\\
	|v\rangle \quad&\mapsto \quad
	\Phi\big( |v\rangle\big)
\end{align*}
between these two linear spaces.
We omit the concrete construction of this isomorphism $\Phi$ here and recommend the reader to the \S 5 of \cite{DJM00}.

Under the boson-fermion correspondence,
a vector $|v\rangle\in\mathcal{F}^{(0)}$ is mapped to a formal power series $\tau_{|v\rangle}(\mathbf{T}):=\Phi(|v\>)\in\mathbb{C}[\![\mathbf{T}]\!]$,
called the bosonic tau-function or tau-function for short.
Moreover,
if $|v\rangle$ is a fermionic tau-function of the KP hierarchy,
i.e. it is of the form \eqref{eqn:fi wedge},
the Pl{\"u}cker relations satisfied by $|v\rangle$ are transferred to the following Hirota quadratic equation satisfied by $\tau_{|v\rangle}(\mathbf{T})$,
\begin{align*}
	\mathop{\Res}_{z=\infty}\,  e^{\xi(\mathbf{T}-\mathbf{T}',z)}
	\tau_{|v\rangle}(\mathbf{T}-[z^{-1}])
	\tau_{|v\rangle}(\mathbf{T}'+[z^{-1}]) dz =0,
\end{align*}
where $\xi(\mathbf{T},z)=\sum_{k=1}^\infty T_k z^k$ and
$$\mathbf{T}\pm[z^{-1}]=(T_1\pm z^{-1},T_2\pm z^{-2}/2,\dots).$$

The affine coordinates  $\{b_{k,i}\}_{k,i=0}^{\infty}$ of a fermionic tau-function are related to certain special evaluations of the corresponding bosonic tau-function (cf. \cite{BBT03,DJM00,Zhou15}).
First, let us introduce the fermionic one-point function and the fermionic dual one-point function:
\begin{align*}
	\Psi(z) := 1+\sum_{i=0}^\infty b_{0,i} z^{-i-1},\qquad
	\Psi^*(z) := 1+\sum_{i=0}^\infty b_{i,0} z^{-i-1},
\end{align*}
Then, via the boson-fermion correspondence,
one has (see \S 6.3 in \cite{DJM00}, and also equations (129) and (130) in \cite{Zhou15}),
\begin{equation}\label{defnone}
	\Psi(z)=\frac{\tau(\mathbf{T})|_{T_k=-z^{-k}/k }}{\tau(\mathbf{T})|_{T_k=0}} ,\qquad
	\Psi^*(z)=\frac{\tau(\mathbf{T})|_{T_k=z^{-k}/k }} {\tau(\mathbf{T})|_{T_k=0}}.
\end{equation}
Second, consider the fermionic two-point function, which is the generating function of all the affine coordinates $\{b_{k,i}\}_{k,i=0}^\infty$ and is defined by:
\begin{align}\label{eqn:Puv as b}
	\Psi(w,z)  :=\frac{1}{w-z} + \sum_{i,j\geq 0} b_{i,j} w^{-i-1} z^{-j-1}
	\in\frac{1}{w-z}+w^{-1}z^{-1}\mathbb{C}[\![w^{-1},z^{-1}]\!].
\end{align}
Then, by the standard method using vertex operators,
one obtains (see, for example, equation (201) in \cite{Zhou15})
\begin{align*}
	\Psi(w,z) 
	=\frac{1}{w-z} \frac{\tau(\mathbf{T})|_{T_k=\frac{w^{-k}-z^{-k}}{k} }}{\tau(\mathbf{T})|_{T_k=0}} .
\end{align*}

The above specialization of the tau-function $\tau(\mathbf{T})$ of the KP hierarchy contains all information about the tau-function $\tau(\mathbf{T})$ up to the non-zero scaling constant $\tau(\mathbf{0})$,
since its expansion \eqref{eqn:Puv as b} yields all the affine coordinates of the corresponding fermionic tau-function.
Explicitly, 
\beq\label{eqn:tau-affine}
\tau({\bf T})=\tau(\mathbf{0}) \sum_{\mu} C_{\mu}\cdot S_{\mu}({\bf T}).
\eeq
Here, the summation is over all partitions $\mu$, $S_{\mu}({\bf T})$ is the Schur polynomial, and $C_{\mu}$ is given by
$$
C_{\mu}=(-1)^{d+\sum_{i=1}^d n_i}\det\big((b_{m_i,n_j})_{i,j=1}^{d}\big)
$$ 
where $(m_1,\dots,m_d|n_1,\dots,n_d)$ is the Frobenius notation of the partition $\mu$.

\subsection{$W_{1+\infty}$ algebra}
In this subsection, we review the $w_{1+\infty}$ algebra and $W_{1+\infty}$ algebra.

The ring of differential operators on $S^1$ is
\begin{align*}
	w_{1+\infty}=\text{span}_{\mathbb{C}}\{z^a\partial_z^b
	|a\in\mathbb{Z}, b\in\mathbb{Z}_{\geq0}\}.
\end{align*}
Elements in $w_{1+\infty}$ can naturally act on the fermionic Fock space $\mathcal{F}^{(0)}$ by
\begin{align*}
	l_z \cdot \big(f_0(z)\wedge f_1\wedge\cdots\big)
	:=\big(l_z\cdot f_0(z)\big)\wedge f_1(z)\wedge\cdots
	+f_0(z)\wedge \big(l_z\cdot f_1(z)\big)\wedge\cdots
	+\cdots
\end{align*}
and its linear expansion,
where $l_z\in w_{1+\infty}$,
and $f_0(z)\wedge f_1(z)\wedge\cdots$
is an element in fermionic Fock space.

As reviewed in the last subsection,
the boson-fermion correspondence gives an isomorphism between the fermionic Fock space $\mathcal{F}^{(0)}$ and the ring $\mathbb{C}[\![\mathbf{T}]\!]$ of formal power series of variables $\mathbf{T}=(T_1,T_2,\dots)$.
Under such isomorphism,
any element $l_z\in w_{1+\infty}$ can be realized as a new operator $\widehat{l_z}$ which acts on $\mathbb{C}[\![\mathbf{T}]\!]$.
This is called the $W_{1+\infty}$ realization of elements in $w_{1+\infty}$.
The concrete construction for $\widehat{l_z}$ can be obtained from the boson-fermionic correspondence and vertex operators
(cf. \cite{FKN92,KS91}).
We do not need the general formula,
so we only list some useful examples as the following
\begin{example}\label{exa:w1+infty}
	$${\renewcommand{\arraystretch}{1.5}
		\begin{array}{|c|c|c|c|}\hline
			\hbox{\small Subalgebras}
			&\hbox{\small $w_{1+\infty}$} 
			& \hbox{\small $W_{1+\infty}$ }\\\hline
			\text{Heisenberg\ algebra}
			& j_n:= -z^{n}, \ n\neq 0
			& J_n=\widehat{j_n} \\\hline
			\text{Virasoro\ algebra}
			&  l_n:= -z^{n}\left(z\frac{\partial}{\partial z} +\frac{n+1}{2}\right)
			& L_n=\widehat{l_n} \\\hline
	\end{array}}$$
	Above operators $J_n$ and $L_n$ are given by
	\begin{align*}
		J_n=\begin{cases}
			-nT_{-n},\ \ &n<0,\\
			\frac{\partial}{\partial T_n},\ \ &n>0,
		\end{cases}
	\end{align*}
	and
	\begin{align*}
		L_n=\begin{cases}
			\sum_{k=1}^{\infty}kT_k\frac{\partial}{\partial T_{n+k}}
			+\frac{1}{2}\sum_{a+b=-n,a,b>0}abT_aT_b,\ \ &n\leq0\\
			\sum_{k=1}^{\infty}(k-n)T_{k-n}\frac{\partial}{\partial T_{k}}
			+\frac{1}{2}\sum_{a+b=n,a,b>0} \frac{\partial^2}{\partial T_a \partial T_b},\ \ &n>0.
		\end{cases}
	\end{align*}
	Since $\{J_n\}$, $\{L_n\}$ satisfies the Heisenberg relations and Virasoro relations,
	they are called the Heisenberg operators and Virasoro operators, respectively.
\end{example}

The advantage of using the notation of $W_{1+\infty}$ realization of elements in $w_{1+\infty}$ is that
it provides an effective method to study the constraints satisfied by certain tau-functions of the KP hierarchy.

\subsection{Kac--Schwarz operator}
In this subsection,
we review the Kac--Schwarz operator for tau-functions of the KP hierarchy.

It is known that Virasoro constraints for the partition function of the 2D quantum gravity model is equivalent to its KdV integrability and string equation~\cite{DVV91,FKN92}.
In \cite{KS91},
Kac and Schwarz gave an explanation of these Virasoro constraints in terms of Sato Grassmannian.
More precisely,
consider the point in the Sato Grassmannian which exactly corresponds to the partition function of the 2D quantum gravity model.
This point can be also regarded as a semi-infinite dimensional subspace $\FV_{\KW}$ of $\FV=z^{1/2}\cdot \mathbb{C}[\![z^{-1}]\!][z]$.
Kac and Schwarz showed that there are some differential operators in $w_{1+\infty}$,
which preserve the space $\FV_{\KW}$,
and Virasoro operators in those Virasoro constraints comes from $W_{1+\infty}$ realization of these differential operators in $w_{1+\infty}$.
This motivates the notion of Kac--Schwarz operator.
Below,
we give a detailed description of this notion.

For a fermionic tau-function $|v\rangle$ of the KP hierarchy,
if its corresponding semi-infinite dimensional subspace in Sato Grassmannian has the following admissible basis
\begin{align*}
	\{\varphi_0(z), \varphi_1(z),\dots\},
\end{align*}
then we say $l_z\in w_{1+\infty}$ is a Kac--Schwarz operator of $|v\rangle$ if and only if
\begin{align*}
	(z^{1/2}\cdot l_z\cdot z^{-1/2}) \varphi_i \in
	\text{Span}_{\mathbb{C}}\{\varphi_0(z), \varphi_1(z),\dots\},\ \ 
	\text{for\ all\ }i\geq0.
\end{align*}
Here the conjugation via $z^{1/2}$ comes from notations,
which is not essential.

Kac--Schwarz operator provides a powerful tool to study constraints for tau-functions of the KP hierarchy.
As we have mentioned in the last subsection,
under the boson-fermion correspondence,
elements in $w_{1+\infty}$ become operators acting on the space $\mathbb{C}[\![\mathbf{T}]\!]$.
When $l_z$ is a Kac--Schwarz operator of a fermionic tau-function $|v\rangle$ of the KP hierarchy,
the $W_{1+\infty}$ realization of $l_z$ provides a constraint for the tau-function $\tau_{|v\rangle}(\mathbf{T})$ (see, for example, Lemma 3.2 in \cite{FKN92}).
More precisely,
\begin{proposition}
	For a fermionic tau-function $|v\rangle$ of the KP hierarchy,
	we denote by $\tau_{|v\rangle}(\mathbf{T})$ the corresponding bosonic tau-function.
	Then $l_z$ is a Kac--Schwarz operator for $|v\rangle$ if and only if there exists a constant $c_{l_z}$ such that
	\begin{align*}
		\widehat{l_z} \cdot \tau_{|v\rangle}(\mathbf{T})
		=c_{l_z} \cdot \tau_{|v\rangle}(\mathbf{T}).
	\end{align*}
\end{proposition}
From now on,
we use the term that $l_z$ is a Kac--Schwarz operator for the tau-function $\tau(\mathbf{T})$ of the KP hierarchy when the action of $\widehat{l_z}$ on $\tau(\mathbf{T})$ is equal to taking the multiplication with a constant.

\subsection{Polynomial reduction}
With the language of Kac--Schwarz operators,
the $r$-KdV reduction for a tau-function $\tau({\bf T})$ of the KP hierarchy is equivalent to say that $z^{r}$ is a Kac--Schwarz operator for $\tau(\mathbf{T})$.
See \cite{DJKM82,KS91} for more details.
Now we generalize this notion to the polynomial-reduced KP hierarchy.
\begin{definition}\label{def:x-red-KP}
	Let $\tau({\bf T})$ be a tau-function of the KP hierarchy, 
	we call $\tau(\bf T)$ a tau-function of the polynomial-reduced KP hierarchy if
	there exists a polynomial $x(z)\in \mathbb C[z]$ such that one of the following two equivalent conditions holds: 
	\begin{enumerate}
		\item the operator $x(z)\cdot$ is a Kac--Schwarz operator for the tau-function $\tau({\bf T})$;
		\item the action of the operator $\widehat{x(z)}$,
		which is the $W_{1+\infty}$ realization of $x(z)$,
		on the tau-function $\tau({\bf T})$ is equal to taking the multiplication with a constant.
	\end{enumerate}
	When the polynomial $x(z)$ is explicitly known, the corresponding hierarchy here is also referred to as the $x$-reduced KP hierarchy.
\end{definition}
Clearly, when $x(z)=z^r$, the $x$-reduced KP hierarchy recovers the $r$-KdV hierarchy.
\begin{remark}
	In \cite{GJZ23}, the notion of a polynomial-reduced KP hierarchy is formulated via Condition (2) in Definition~\ref{def:x-red-KP}, with the additional requirement that all the resulting constants vanish.
	In our formulation we do not impose this vanishing condition.
	However, the distinction is inessential: if the action of $\widehat{x(z)}$ on a tau-function $\tau({\bf T})$ produces nonzero constants, one can always multiply $\tau({\bf T})$ by an exponential factor $e^{f({\bf T})}$ with $f({\bf T})$ linear in ${\bf T}$, so that the modified tau-function satisfies the normalization used in \cite{GJZ23}.
	Thus the two definitions are equivalent up to such a normalization of the tau-function.
	See also \cite{ABDKS25,CGG25,GJZ23} for several examples of $x$-reduced KP hierarchies arising from topological recursion of spectral curves.
\end{remark}

In general, an $x$-reduced KP hierarchy together with an appropriate string equation determines a tau-function uniquely up to a nonzero multiplicative constant. 
Classical examples include the partition function of 2D topological gravity, which is characterized by the KdV hierarchy and the string equation (see \cite{KS91,Wit91}), 
and the partition function of the $r$-spin theory, which is characterized by the $r$-KdV hierarchy and the corresponding string equation (see \cite{ACM12,AM92,FKN92,Wit93}). 
For further discussion, we refer the reader to the proof and remark in Theorem 2.3 of \cite{ACM12} and to Lemma 3.6 of \cite{Ale21a}.
Consequently, in order to establish the equivalence between two models, it suffices to verify that they satisfy the same $x$-reduced KP hierarchy, the same string equation, and the same initial condition. 
This strategy is employed in the present work to prove the equivalence between the GKM, the topological recursion of the shifted $r$-spin spectral curve and the geometry of $r$-spin theory.

\section{Generalized Kontsevich model and its KP integrability}
\label{sec:GKMM}
In this section,
we review the definition of the GKM with a polynomial potential,
and its relation to the KP hierarchy.
We also derive a differential version of the string equation of the GKM,
which plays an essential role in deriving the corresponding spectral curve.

\subsection{Generalized Kontsevich model}
\label{sec:gKmm def}
The famous Witten Conjecture \cite{Wit91} states that the generating function of intersection numbers over moduli spaces of stable curves is a tau-function of the KdV hierarchy and satisfies the string equation.
In \cite{Kont92},
Kontsevich introduced the Kontsevich matrix model to capture the properties of this tau-function.
The Kontsevich matrix model uses the matrix Airy function and is of the following form
\begin{align}\label{eqn:def KW}
	\tau_{\KW}([\Lambda]^{-1})
	=\frac{1}{\mathcal{C}_{\KW}}
	\int [d\Phi]
	\exp\Big(-\frac{1}{\hbar}
	\text{Tr}\big(\Phi^3/6+\Lambda\Phi^2/2\big)\Big),
\end{align}
where the integral is over the space of $N\times N$ Hermitian matrices,
$\Lambda=\diag\{\lambda_1,\dots,\lambda_N\}$ is a diagonal matrix, $[d\Phi]$ is the Lebesgue measure,
and $\mathcal{C}_{\KW}$ is the normalization constant of this matrix model.
Under taking $N\rightarrow +\infty$ and using the Miwa variables $T_k=\text{Tr}\ \Lambda^{-k}/k$ for $k\in\mathbb{Z}_{\geq1}$,
the above partition function $\tau_{\KW}$ produces the generating function of intersection numbers over moduli spaces of stable curves.

In the Kontsevich matrix model,
a cubic function is used in the integrand as in the right hand side of equation \eqref{eqn:def KW}.
Later, this setup is generalized to general polynomials,
which play an important role in studying the two dimensional quantum gravity and the moduli spaces of $r$-spin curves.
See \cite{AM92,Ale21a,FKN92,GJY25,KMMMZ} and references therein.
We follow the notations in \cite{Ale21a,KMMMZ} to introduce this model.
First,
for a given polynomial $V(z)$ whose degree is greater than 2,
the associated GKM is defined by
\begin{align}\label{eqn:def Z_V}
	Z_{V,N}([\Lambda]^{-1})=\frac{1}{\mathcal{C}_V}
	\int [d\Phi] \exp\Big(-\frac{1}{\hbar}
	\text{Tr}\big(V(\Phi)-\Phi V'(\Lambda)\big)\Big),
\end{align}
where the integral is again over the space of $N\times N$ Hermitian matrices,
$[d\Phi]$ is the standard Lebesgue measure,
and $\mathcal{C}_V$ is the normalization constant given by
\begin{align*}
	\mathcal{C}_V:=\hbar^{\frac{N(N-1)}{2}}
	\frac{\triangle(\lambda)}{\triangle\big(V'(\lambda)\big)}
	\sqrt{\det\big(2\pi\hbar/V''(\Lambda)\big)}
	e^{\text{Tr}(\Lambda V'(\Lambda)-V(\Lambda))/\hbar}.
\end{align*}
Here $\Lambda=\diag\{\lambda_1,\dots,\lambda_N\}$ is a diagonal matrix,
and $\triangle(\lambda)$ means the Vandermonde determinant given by $\prod_{1\leq i<j\leq N} (\lambda_j-\lambda_i)$.
Below,
we give a detailed description of the integral in the right hand side of equation \eqref{eqn:def Z_V}
since it seems to be divergent at the first glance.
For convenience we set $V(z)=\sum_{i}v_iz^i$.
Then after taking the change of variable $\Phi\rightarrow\Phi+\Lambda$,
the function $\text{Tr}\big(V(\Phi)-\Phi V'(\Lambda)\big)$ in the integrand becomes
$$
	V(\Lambda)-\Lambda V'(\Lambda)
	+\sum_{i}\frac{i v_i}{2} \sum_{ a+b=i-2;\,  a, b\geq0}
	\text{Tr}(\Lambda^a \Phi \Lambda^b \Phi)
	+\cdots,
$$
where the omitted terms are at least third-order of $\Phi$.
Thus the leading term of the matrix integral \eqref{eqn:def Z_V} becomes Gaussian,
which is obviously well-defined.
In particular,
when $V(z)=z^3/6$ and after the change of variable $\Phi\rightarrow\Phi+\Lambda$,
equation \eqref{eqn:def Z_V} reduces to equation \eqref{eqn:def KW},
i.e. the GKM generalizes the Kontsevich matrix model.

Similar to the Kontsevich matrix model,
the partition function of the GKM $Z_V([\Lambda]^{-1})$ is a symmetric function in variables $\{\lambda_1^{-1},\dots,\lambda_N^{-1}\}$.
Thus, one can take $N\rightarrow+\infty$ and use the Miwa variables $T_k=\text{Tr}\Lambda^{-k}/k$ for $k\geq1$,
to obtain a formal power series of variables $\mathbf{T}=(T_1,T_2,\dots)$,
which we denote it by $Z_V(\mathbf{T};\hbar)$.
We call $Z_V(\mathbf{T};\hbar)$ the partition function of the GKM associated with the potential $V(z)$.

To calculate the correlation functions of the GKM,
one can use the famous Harish-Chandra--Itzykson--Zuber formula (see \cite{HC57, IZ80}).
The result is that equation \eqref{eqn:def Z_V} can be expressed as
\begin{align}\label{eqn:Z_V as det}
	Z_V([\Lambda]^{-1})=\frac{\det (\Phi^V_i(\lambda_j))_{i,j=1}^N}
	{\triangle(\lambda)},
\end{align}
where the denominator is the Vandermonde determinant,
and the functions appeared in the numerator are given by the integral
\begin{align}\label{eqn:Phi^V_k}
	\Phi^V_k(z)
	=\sqrt\frac{V''(z)}{2\pi\hbar}
	e^{\frac{1}{\hbar}\big(V(z)-zV'(z)\big)}
	\int_{\gamma(V)} d\phi\ \phi^{k-1} 
	e^{-\frac{1}{\hbar}\big(V(\phi)-\phi V'(z)\big)},
\end{align}
computed by steepest descent method at $|z|\rightarrow \infty$. 
The contour $\gamma(V)$ is chosen as the (local) path of steepest descent passing through the saddle point $\phi_0=z$, whose direction is consistent with the branch choice of $\sqrt{V''(z)}$. 
One can also consider an equivalent approach to understand the integral,  similar to the discussion which rigorously defines the GKM.
Under the change of variable $\phi\rightarrow\phi+z$,
the integrand $\phi^{k-1} 
e^{-\frac{1}{\hbar}\big(V(\phi)-\phi V'(z)\big)}$ becomes
\begin{align*}
	(\phi+z)^{k-1} 
	\cdot \exp\Big(-\frac{1}{\hbar}\big(V(z)-zV'(z)+
	\frac{V''(z)}{2}\phi^2
	+\sum_{i\geq3} \frac{V^{(i)}(z)}{i!}\phi^i\big)\Big).
\end{align*}
One can notice that the leading term is quadratic and higher order term can be dealt by perturbation expansion.
Thus,
the integral in the right hand side of equation \eqref{eqn:Phi^V_k} is defined in terms of Gaussian integral,
and as a consequence,
\begin{align*}
	\Phi^V_k(z)
	\in z^{k-1} \cdot \mathbb{C}[\![z^{-1}]\!],
	\,\, k=1,2,\dots.
\end{align*}

This kind of determinantal formula, as equation \eqref{eqn:Z_V as det},
is deeply related to the Sato's theory of KP hierarchy.
See, for example, Lemma 4.2 in \cite{Kont92}.
The result is,
under the Miwa variables $T_k=\text{Tr}\ \Lambda^{-k}/k$ for $k\geq1$,
the partition function $Z_V(\mathbf{T};\hbar)$ is a tau-function of the KP hierarchy.
For the case of $N=1$ in equation \eqref{eqn:Z_V as det},
we have $Z_V(T_k=\lambda_1^{-1}/k)=\Phi^V_1(\lambda_1)$.
Comparing this to the second part of equation \eqref{defnone},
we know that $\Phi^V_1(z)$ gives the fermionic dual one-point function of this model.
The general functions $\Phi^V_k(z), k=1,2,\dots$ are actually dual basis vectors for fermionic tau-functions of the KP hierarchy.
Further,
for this model we set $\Psi^V_k(z)=\Phi^V_k(z)|_{\hbar\rightarrow-\hbar}$,
i.e.,
\begin{align}\label{eqn:Psi^V_k}
	\Psi^V_k(z)
	=\sqrt\frac{V''(z)}{-2\pi\hbar}
	e^{-\frac{1}{\hbar}\big(V(z)-zV'(z)\big)}
	\int_{\gamma(V)} d\phi\ \phi^{k-1} 
	e^{\frac{1}{\hbar}\big(V(\phi)-\phi V'(z)\big)}.
\end{align}
Then it was known (see, for examples, \cite{AM92,Ale21a,KMMMZ}) that
the partition function of the GKM $Z_V(\mathbf{T};\hbar)$ is a tau-function of the KP hierarchy in Miwa variables $\mathbf{T}=(T_1,T_2,\dots)$,
and $\{z^{1/2}\Psi^V_{k+1}(z)\}_{k=0}^\infty$ is an admissible basis for it.
See subsection \ref{subsec:fermionic tau} for a review for the notion.

\subsection{Kac--Schwarz operators for the partition function of generalized Kontsevich model}
As a tau-function of the KP hierarchy,
the partition function of the GKM is determined by its admissible basis $\{z^{1/2} \Psi^V_{k+1}(z)\}_{k=0}^{\infty}$ when considering the corresponding point in Sato Grassmannian.
Let $x(z)=V'(z)=\sum_{i=0}^r a_i z^i$ and
\begin{align}
	\label{eqn:KS Q for GKM}
	Q_V(z)=\text{Ad}_{\sqrt{V''(z)}\, e^{-\frac{1}{\hbar}\left(V(z)-zV'(z)\right)}}
	\Big(-\frac{\hbar}{x'(z)} \partial_z\Big)
	=z-\frac{\hbar}{x'(z)} \partial z +\frac{\hbar}{2}\frac{x''(z)}{x'(z)^2},
\end{align}
then by the integral expression \eqref{eqn:Psi^V_k} of $\Psi^V_k(z)$ and the method of integration by parts,
one has 
\begin{align}
	&\big(z^{1/2}\cdot Q_V(z)\cdot z^{-1/2}\big)
	\cdot z^{1/2} \Psi^V_k(z)
	=z^{1/2} \Psi^V_{k+1}(z), \label{eqn:QV act Psi}\\
	&\big(z^{1/2}\cdot x(z)\cdot z^{-1/2}\big)
	\cdot z^{1/2} \Psi^V_k(z)
	=z^{1/2} \Big(\sum_{i=0}^{r} a_i\Psi^V_{k+i}(z)+\hbar(k-1)\Psi^V_{k-1}(z)\Big).
\end{align}
Thus, directly from the definition of the Kac--Schwarz operator,
the operators $Q_V(z)$ and $x(z)\cdot$ are Kac--Schwarz operators for the partition function of the GKM.
See also \cite{Ale21a} for the derivation of above Kac--Schwarz operators.
We remark that our notation is slightly different from the one in~\cite{Ale21a}, and our basis vectors are his dual basis vectors.

Recall that the fermionic one-point and dual one-point functions of a tau-function of the KP hierarchy are formal power series of the following form
\[\Psi(z)=1+z^{-1}\mathbb{C}[\![z^{-1}]\!],\quad\quad \Psi^{*}(w)\in1+w^{-1}\mathbb{C}[\![w^{-1}]\!].\]
For the partition function of the GKM,
we denote them by $\Psi_V(z)$ and $\Psi^*_V(w)$.
These two functions are equal to $\Psi^V_1(z)$ and $\Psi^V_1(w)|_{\hbar\rightarrow-\hbar}$, respectively.
They can be determined by the following equations
\begin{align}
	&\left(x\big(Q_V(z)\big)-x(z)\right)\, \Psi_V(z)=0, \label{eqn:qsc}\\
	&\left(x\big(Q^*_V(w)\big)-x(w)\right)\, \Psi^{*}_V(w)=0, \label{eqn:qsc dual}
\end{align}
where $Q_V^*(z)=Q_{-V}(z)$ is the dual operator of $Q_V(z)$.
To specify a tau-function of the KP hierarchy,
we need the fermionic two-point function of it,
which is the generating function of affine coordinates as in equation \eqref{eqn:Puv as b}.
In \cite{GJY25},
the first three authors studied the lifting operators of the KP hierarchy.
As an application,
a compact formula for the fermionic two-point function of the GKM is obtained as follows.
\begin{proposition}[Proposition 5.2 in \cite{GJY25}]
	\label{pro:fermionic two point}
	The fermionic two-point function for the GKM has the following explicit formula
	\begin{align}\label{eqn:2-point formula}
		\Psi_V(w,z)=\frac{W\big(Q_V^*(w),Q_V(z)\big)}{x(w)-x(z)}
		\big(\Psi^{*}_V(w)\Psi_V(z)\big)
		\in\frac{1}{w-z}+w^{-1}z^{-1}\mathbb{C}[\![w^{-1},z^{-1}]\!],
	\end{align}
	where $W(w,z)=\frac{x(w)-x(z)}{w-z}\in\mathbb{C}[w,z]$
	is a polynomial with indeterminates $w$ and $z$.
\end{proposition}

Together with equations~\eqref{eqn:qsc} and \eqref{eqn:qsc dual}, this Proposition uniquely determines the fermionic two-point function (thus the affine coordinates) of the GKM from its potential $V(z)$ (more precisely, from $x(z)=V'(z)$).
Furthermore, by using equation~\eqref{eqn:tau-affine}, one can compute the partition function $Z_{V}({\bf T};\hbar)$ explicitly.
We show some explicit computations for the partition function of the GKM with potential $V(z)=\frac{z^4}{12}-\frac{\epsilon z^2}{2}$.
\begin{example}\label{exam:partition-function-GKM}
	When $V(z)=\frac{z^4}{12}-\frac{\epsilon z^2}{2}$,
	we have $x(z)=\frac{z^3}{3}-\epsilon z$, and
	\begin{align*}
		Z_{V}({\bf T};\hbar)
		=&\, \textstyle 
		1+ {\hbar}\,{T_{{1}}}^{2}T_{{2}}+{\frac {{\hbar}\,}{3}} T_{{4}}
		+  {\frac {4\,{\hbar}\,\epsilon}{3}}{T_{{2}}}^{3}
		+2\,{\hbar}\,\epsilon T_{{6}}+2\,{\hbar}\,\epsilon {T_{{1}}}^{2}T_{{4}}
		+6\,{\hbar}\,\epsilon T_{{1}}T_{{2}}T_{{3}}  
		+  {\frac {{{\hbar}}^{2}}{2}}{T_{{1}}}^{4}{T_{{2}}}^{2}\\
		&\, \textstyle 
		+{\frac {5\,{{\hbar}}^{2}}{3}}{T_{{1}}}^{3}T_{{5}}
		+{\frac {7\,{{\hbar}}^{2}}{3}}T_{{1}}T_{{7}}
		+6\,{\hbar}\,{\epsilon}^{2}T_{{8}}
		+3\,{\hbar}\,{\epsilon}^{2}{T_{{1}}}^{2}T_{{6}}
		+8\,{\hbar}\,{\epsilon}^{2}{T_{{2}}}^{2}T_{{4}}
		+9\,{\hbar}\,{\epsilon}^{2}T_{{2}}{T_{{3}}}^{2}\\
		&\, \textstyle
		+{\frac {13\,{{\hbar}}^{2}}{3}}{T_{{1}}}^{2}T_{{2}}T_{{4}}
		+10\,{\hbar}\,{\epsilon}^{2}T_{{1}}T_{{2}}T_{{5}}
		+12\,{\hbar}\,{\epsilon}^{2}T_{{1}}T_{{3}}T_{{4}}
		-{\frac {2\,{{\hbar}}^{2}}{3}}{T_{{2}}}^{4}
		+{\frac {13\,{{\hbar}}^{2}}{18}} {T_{{4}}}^{2}
		+\cdots
	\end{align*}
\end{example}

\subsection{Connected $n$-point functions and differential $n$-forms}
\label{sec:diff version}
Now we introduce the connected $n$-correlators with genus $g$:
$$
	\langle \alpha_{i_1},...,\alpha_{i_n} \rangle^{V}_{g,n}
	:=[\hbar^{2g-2+n}] 
	\left.\frac{\partial^n \log Z_V(\mathbf{T};\hbar)}
	{\partial T_{i_1} \dots \partial T_{i_n}}
	\right|_{\mathbf{T}=0},
$$
and the associated connected $n$-point function with genus $g$ for the GKM:
$$
	H^V_{g,n}(z_1,...,z_n):=\sum_{i_1,...,i_n=1}^{\infty}
	\langle \alpha_{i_1},...,\alpha_{i_n} \rangle^{V}_{g,n} z_1^{-i_1-1} \dots z_n^{-i_n-1}.
$$
Notice that the function $H^V_{g,n}(z_1,...,z_n)$ is a generating function of all the connected $n$-correlators with genus $g$.
We also introduce the differential version of the connected $n$-point function with genus $g$ of the GKM as
\beq\label{eqn:def omega^V}
	\omega^{V}_{g,n}(z_1,...,z_n):=
	\sum_{i_1,...,i_n=1}^\infty
	\frac{1}{i_1\cdots i_n}\langle \alpha_{i_1},...,\alpha_{i_n} \rangle^{V}_{g,n}
	d z_1^{-i_1} \cdots d z_n^{-i_n}.
\eeq
Clearly, the relation between the $n$-differential above and the ordinary connected $n$-point function with genus $g$ is
$$
	\omega^{V}_{g,n}(z_1,...,z_n)
	=(-1)^n H^{V}_{g,n}(z_1,...,z_n)dz_1\cdots dz_n.
$$
We note here that both the function $H^{V}_{g,n}(z_1,\cdots,z_n)$ and the multi differential $\omega^{V}_{g,n}(z_1,\cdots,z_n)$ are defined near $z_1,\cdots,z_n=\infty$. 

It is easy to see that,
the function $H^{V}_{g,n}$ and its corresponding differential $\omega^{V}_{g,n}$ contain the same information.
There are many methods to compute such $n$-point functions (see, for examples, \cite{BE09,DYZ21,Okou02,WY22,Zhou15} and reference therein).
In this paper, according to~\cite[equation (211) and Theorem 5.3]{Zhou15}, the connected $n$-point functions can be computed by using the fermionic two-point function $\Psi_V(w,z)$ via the following formula:
$$
\sum_{g=0}^\infty H_{g,n}^V(z_1,\dots,z_n) \hbar^{2g-2+n}
=(-1)^{n-1}\sum_{\text{$n$-cycles}\ \sigma}
\prod_{i=1}^n \widetilde{\Psi}_V(z_{\sigma^i(1)},z_{\sigma^{i+1}(1)})
-i_{z_1,z_2}\frac{\delta_{n,2}}{(z_1-z_2)^2},
$$
where the notation $i_{z,w}f(z,w)$ means we expand the function $f$ in the region $|z|>|w|$ and
\begin{align*}
	\widetilde{\Psi}_V(z_{i},z_{j}):=\begin{cases}
		-\Psi_V(z_{i},z_{j})+\frac{1}{z_{i}-z_{j}}+i_{z_i,z_j}\frac{1}{z_i-z_j}, & \text{if\ }i<j,\\
		-\Psi_V(z_{i},z_{j})+\frac{1}{z_{i}-z_{j}}, & \text{if\ }i=j,\\
		-\Psi_V(z_{i},z_{j})+\frac{1}{z_{i}-z_{j}}+i_{z_j,z_i}\frac{1}{z_i-z_j}, & \text{if\ }i>j.
	\end{cases}
\end{align*}

We show some explicit computations for the connected multi-differentials of the GKM with potential $V(z)=\frac{z^4}{12}-\frac{\epsilon z^2}{2}$, and one can compare these computations with the computation of the partition function shown in Example~\ref{exam:partition-function-GKM}.

\begin{example}
	When $V(z)=\frac{z^4}{12}-\frac{\epsilon z^2}{2}$,
	we have $x(z)=\frac{z^3}{3}-\epsilon z$.
	Then for $n=1$ we have
	\begin{align*}
		\sum_{g=0}^\infty \omega^V_{g,1}(z_1) \cdot \hbar^{2g-1}
		=&-\bigg(\frac{\hbar}{3}\frac{dz_1}{z_1^5}
		+2 \epsilon \hbar\frac{dz_1}{z_1^7}
		+6 \epsilon^2 \hbar\frac{dz_1}{z_1^9}
		+\frac{40 \epsilon^3 \hbar}{3} \frac{dz_1}{z_1^{11}}
		+25\epsilon^4\hbar\frac{dz_1}{z_1^{13}} \\
		&+42\epsilon^5\hbar\frac{dz_1}{z_1^{15}}
		 +\frac{385 \epsilon \hbar^3}{9} \frac{dz_1}{z_1^{15}}
		+{\frac {196 \,{\epsilon}^{6}\hbar}{3}}\frac{dz_1}{{z_{{1}}}^{17}}
		+\frac{4340 \epsilon^2 \hbar^3}{9} \frac{dz_1}{z_1^{17}}
\bigg)
		+O(z_1^{-19}).
	\end{align*}
	For $n=2$, we have
	\begin{align*}
		\sum_{g=0}^\infty \omega^V_{g,2}(z_1,z_2) \cdot &\hbar^{2g}
		=\frac{7\hbar^2}{3}\frac{dz_1dz_2}{z_1^{2}z_2^{8}}
		+\frac{4\hbar^2}{3}\frac{dz_1dz_2}{z_1^{5}z_2^{5}}
		+\frac{7\hbar^2}{3}\frac{dz_1dz_2}{z_1^{8}z_2^{2}}
		+\frac{40\hbar^2}{3}\frac{dz_1dz_2}{z_1^{3}z_2^{9}}
		+7\hbar^2\epsilon\frac{dz_1dz_2}{z_1^{4}z_2^{8}}\\
		&+8\hbar^2\epsilon\frac{dz_1dz_2}{z_1^{5}z_2^{7}}
		+\frac{25\hbar^2\epsilon}{3}\frac{dz_1dz_2}{z_1^{6}z_2^{6}}
		+8\hbar^2\epsilon\frac{dz_1dz_2}{z_1^{7}z_2^{5}}
		+7\hbar^2\epsilon\frac{dz_1dz_2}{z_1^{8}z_2^{4}}
		+\frac{40\hbar^2}{3}\frac{dz_1dz_2}{z_1^{9}z_2^{3}}\\
		&+\frac{152\hbar^2\epsilon^2}{3}\frac{dz_1dz_2}{z_1^{5}z_2^{9}}
		+\frac{140\hbar^2\epsilon^2}{3}\frac{dz_1dz_2}{z_1^{6}z_2^{8}}
		+48\hbar^2\epsilon^2\frac{dz_1dz_2}{z_1^{7}z_2^{7}}
		+\frac{140\hbar^2\epsilon^2}{3}\frac{dz_1dz_2}{z_1^{8}z_2^{6}}\\
		&+\frac{152\hbar^2\epsilon^2}{3}\frac{dz_1dz_2}{z_1^{9}z_2^{5}}
		+184\hbar^2\epsilon^3\frac{dz_1dz_2}{z_1^{7}z_2^{9}}
		+\frac{539\hbar^2\epsilon^3}{3}\frac{dz_1dz_2}{z_1^{8}z_2^{8}}
		+184\hbar^2\epsilon^3\frac{dz_1dz_2}{z_1^{9}z_2^{7}}\\
		&-\frac{680\hbar^4}{9}\frac{dz_1dz_2}{z_1^{9}z_2^{9}}
		-\frac{359\hbar2\epsilon^4}{6}\frac{dz_1dz_2}{z_1^{9}z_2^{9}}
		+O(z_1^{-10})+O(z_2^{-10}).
	\end{align*}
\end{example}

\subsection{String equation for generalized Kontsevich model and its differential version}
In this subsection, we first consider the $W_{1+\infty}$ realization of the operator $Q_V$,
which provides the string equation for the GKM.
Then we derive the differential version of the string equation.

\begin{proposition}\label{lem:w1+infty of qv}
	We suppose $V(z)$ is a polynomial of degree $r+1$,
	and the expansion $\frac{1}{x'(z)}=\sum_{k=r}^{+\infty} c_k z^{-k+1}$.
	Then the $W_{1+\infty}$ realization of the operator
	$$
		Q_V(z)=z-\frac{\hbar}{x'(z)} \partial z +\frac{\hbar}{2}\frac{x''(z)}{x'(z)^2}
	$$
	is given by
	$$
		\widehat{Q_V}=-\frac{\partial}{\partial T_1}
		+\sum_{k=r}^{+\infty} c_k L_{-k},
	$$
	where $L_{-k}$ is the $W_{1+\infty}$ realization of $-z^{-k}\big(z\partial_z+\frac{-k+1}{2}\big)$, which is precisely given by
	\begin{align*}
		L_{-k}
		=\sum_{m\geq1}(m+k)T_{m+k}\frac{\partial}{\partial T_m}
		+\frac{1}{2}\sum_{a+b=k \atop a,b\geq1}abT_aT_b,
		\ \ k\geq0.
	\end{align*}
\end{proposition}
\begin{proof}
The proof follows instantly rom examples of $W_{1+\infty}$ realization in Example \ref{exa:w1+infty} and direct computations.
\end{proof}

From equation \eqref{eqn:QV act Psi},
we have known that $Q_V(z)$ is a Kac--Schwarz operator for the partition function of the GKM.
In \cite{Ale21a},
it is further proved that
\begin{align}\label{eqn:string for GKM}
	\widehat{Q_V} Z_V({\bf T};\hbar)=0,
\end{align}
which is called the string equation for the GKM.
Moreover,
the partition function $Z_{V}({\bf T};\hbar)$ is a tau-function of the $x$-reduced KP hierarchy and is uniquely determined by the above string equation. 

For further use, we establish the following differential version of the string equation \eqref{eqn:string for GKM} for the GKM.
\begin{proposition}\label{prop:GKM-string-wgn}
	The string equation for the GKM
	$$\widehat{Q_{V}}Z_{V}({\bf T};\hbar)=0$$
	is equivalent to the following equation for multi-differentials $\omega^{V}_{g,n}$:
\begin{align}
	&\, \mathop{\Res}_{z_0=\infty} z_0\cdot \omega^{V}_{g,n+1}(z_0,z_{[n]})
	=-\sum_{j=1}^n d\Big( \frac{\omega^{V}_{g,n}(z_{[n]})}{dx(z_j)} \Big)
	+\delta_{(g,n)}^{(0,2)}\,d_1d_2\frac{x'(z_1)-x'(z_2)}{x'(z_1)x'(z_2)(z_1-z_2)},
	\label{eqn:GKM-string-wgn}
\end{align}
where we recall $x(z)=V'(z)$.
\end{proposition}
\begin{proof} 
	It follows from the definition that $\widehat{Q_{V}}Z_{V}({\bf T};\hbar)/Z_{V}({\bf T};\hbar)\in \mathbb Q[[{\bf T};\hbar]]$. By considering the coefficient of $T_{k_1}\cdots T_{k_n}\cdot\hbar^{2g-2+n}$ in this expression,
	we see the string equation is equivalent to the following system of equations for the connected correlators:
	for $2g-2+n+1>0$,
	\begin{align*}
		\<\alpha_1,\alpha_{k_1},...,\alpha_{k_n}\>^{V}_{g,n+1}
		=\delta_{(g,n)}^{(0,2)}k_1k_2c_{k_1+k_2}
		+\sum_{m=r}^{+\infty} c_m \sum_{j=1}^n k_j
		\<\alpha_{k_1},,...,\alpha_{k_j-m},...,\alpha_{k_n}\>^{V}_{g,n},
	\end{align*}
	where $\<-\>^{V}_{g,n}:=0$ if it contains insertion $\alpha_{-k}$ for $k\geq 0$.
	By multiplying $\frac{dz_1^{-k_1}}{k_1}\cdots \frac{dz_n^{-k_n}}{k_n}$ and then taking summation over $k_1,\cdots,k_n=1,2,\cdots$, these equations are equivalent to
	\begin{align}\label{eqn:Resz0=c}
		\mathop{\Res}_{z_0=\infty} z_0\cdot \omega^{V}_{g,n+1}(z_0,z_{[n]})
		=\delta_{(g,n)}^{(0,2)}\sum_{k_1+k_2\geq r}c_{k_1+k_2}dz_1^{-k_1}dz_2^{-k_2}
		-\sum_{j=1}^n d\Big( \frac{\omega^{V}_{g,n}(z_{[n]})}{dx(z_j)} \Big),
	\end{align}
	where we have used
	$$
	\sum_{m\geq r}\sum_{k_j>m}c_m\cdot \alpha_{k_j-m}\cdot dz_j^{-k_j}
	=\sum_{k_j\geq 1}\alpha_{k_j}\sum_{m\geq r}c_{m}\cdot dz_j^{-k_j-m}
	=-d\circ\frac{1}{dx(z_j)}\sum_{k_j\geq 1}\alpha_{k_j}\cdot\frac{dz_j^{-k_j}}{k_j}.
	$$
	Furthermore, by noticing
	$$
	\sum_{k_1+k_2\geq r}c_{k_1+k_2}dz_1^{-k_1}dz_2^{-k_2}=-d_1d_2\sum_{m\geq r}c_m\frac{z_1^{-m+1}-z_2^{-m+1}}{z_1-z_2}
	=d_1d_2\frac{x'(z_1)-x'(z_2)}{x'(z_1)x'(z_2)(z_1-z_2)},
	$$
	and substituting it to equation \eqref{eqn:Resz0=c},
	we conclude that the string equation is equivalent to equation~\eqref{eqn:GKM-string-wgn}.
	The Proposition is proved.
\end{proof}

\section{Topological recursion for generalized Kontsevich model}
\label{sec:TR-GKM}
In this section, we prove that for a generic polynomial $V(z)$ of degree $r+1$,
the multi-differentials $\omega^{V}_{g,n}$ of the GKM with potential $V(z)$ 
coincide with the ones defined by the topological recursion of the following shifted $r$-spin curve:
$$
\cC=\big(\mathbb P^1,\quad x(z)=V'(z), \quad y(z)=z\big).
$$
As a corollary, we prove the Gukov--Su{\l}kowski Conjecture for the shifted $r$-spin curve.
\subsection{Topological recursion of spectral curve and its integrability}\label{sec:TR-CohFT-KP}
We first review the Eynard--Orantin topological recursion~\cite{EO07}.
The input of the topological recursion is a set of spectral curve data:
\begin{align*}
	\cC=\big(\Sigma,\ \  x(z),\ \  y(z) \big)
\end{align*}
where $\Sigma$ is a Torelli marked Riemann surface of genus $g$, $x(z)$ and $y(z)$ are two (possibly multi-valued) functions over $\Sigma$ satisfying that: 
\begin{itemize}
\item $dx(z)$ is meromorphic differential with finitely many simple zeroes $\{z^\beta\}_{\beta=1}^{N}$;
\item $dy(z)$ is meromorphic and non-vanishing at the zeros of $dx(z)$.
\end{itemize}
A set of multi-differentials $\{\omega_{g,n}\}_{g,n\geq 0}$ are defined as follows. The initial data are directly given by spectral curve data:
\begin{align*}
	\omega_{0,1}(z_1):=y(z_1)dx(z_1),\qquad
	\omega_{0,2}(z_1,z_2):=B(z_1,z_2),
\end{align*}
where $B(z_1,z_2)$ is the fundamental normalized differential of the second kind on $\Sigma$, also referred to as the Bergman kernel. It is the unique meromorphic bidifferential with a double pole along the diagonal and vanishing A-periods with respect to the Torelli marking. 
Then we introduce the recursion kernel. For each zero $z^\beta$ of $dx$, let $\bar{z}\in\Sigma$ be the local involution of $z$ near the zero $z^\beta$ of $dx(z)$. The recursion kernel reads:
\[
K_{\beta}(z_0,z):=\frac{1}{2}\frac{\int_{z'=\bar{z}}^z \omega_{0,2}(z_0,z')}{(y(z)-y(\bar{z}))dx(z)}.
\]
Then for $2g-2+n+1>0$, $\omega_{g,n}$ are obtained by the following topological recursion:
\begin{align}
	\omega_{g,n+1}(z_0,z_{[n]})
	:=&\, \sum\limits_{\beta}\mathop{\Res}_{z=z^\beta}K_\beta(z_0,z)\bigg(\omega_{g-1,n+2}(z,\bar z,z_{[n]}) \nonumber\\
	&\, +\sum^{\prime}_{\substack{g_1+g_2=g\\ I\bigsqcup J=[n]}}\omega_{g_1,|I|+1}(z,z_{I})\omega_{g_2,|J|+1}(\bar z,z_{J})\bigg),
	\label{eqn:def-TR}
\end{align}
where $\sum\limits^{\prime}$ means we delete the terms containing $w_{0,1}$.
The recursion also gives, for $g\geq 2$,
$$
\omega_{g,0}:=\frac{1}{2-2g}\sum_{\beta}\mathop{\Res}_{z=z^\beta}\omega_{g,1}(z)\int_{z'=z^\beta}^{z}\omega_{0,1}(z').
$$
The terms of $\omega_{0,0}$ and $\omega_{1,0}$ are also defined in~\cite{EO07},  whose definitions are omitted as they are not needed for our study.

It follows from~\cite{EO07} that the multi-differentials $\omega_{g,n}$ are symmetric and meromorphic on $\Sigma^{n}$.
Moreover, for $2g-2+n>0$, the poles of $\omega_{g,n}$ on each copy of $\Sigma$ are located only at the critical points $\{z^\beta\}_{\beta=1}^{N}$.
Introduce a system of the differentials $\{d\zeta^{\bar\beta}_k(z)\}_{\beta=1,2,\cdots,N;k\geq 0}$ by
$$
d\zeta_0^{\bar\beta}(z):=-\frac{1}{\sqrt{x''(z^{\beta})}}\mathop{\Res}_{z'=z^\beta}\frac{B(z,z')}{z'-z^\beta},\qquad
d\zeta_{k}^{\bar\beta}(z):=\Big(-d\circ\frac{1}{dx(z)}\Big)^kd\zeta_0^{\bar\beta}(z).
$$ 
then according to~\cite{DOSS14,Eyn14}, the multi-differentials $\omega_{g,n}$ admit the following expansion:
\beq\label{eqn:TR-CohFT}
\omega_{g,n}(z_1,\cdots,z_n)=\sum_{\substack{k_1,\cdots,k_n\geq 0 \\ 1\leq \beta_1,\cdots,\beta_n\leq N}}
\<\bar e_{\beta_1}\bar\psi^{k_1},\cdots, \bar e_{\beta_n}\bar\psi^{k_n}\>_{g,n}
d\zeta_{k_1}^{\bar \beta_1}(z_1)\cdots d\zeta_{k_n}^{\bar \beta_n}(z_n).
\eeq
The summation on the right-hand side of~\eqref{eqn:TR-CohFT} is finite, since the order of the poles of 
$\omega_{g,n}$ is bounded \cite{EO07}.
In fact, the coefficients $\<-\>_{g,n}$ coincide with the ancestor correlators of the CohFT associated with the spectral curve. We refer the reader to~\cite{DOSS14,Eyn14} for further details; see also~\cite{GJZ23}.

Near a boundary point $b$ (by definition, $b$ is a point on $\Sigma$ such that $x(b)=\infty$), we choose a local coordinate $\lambda$ such that $\lambda^{-1}(b)=0$.
Then near the point $b$, the multi-differentials $\omega_{g,n}$ can be expanded under basis $\{d\lambda^{-k}\}_{k=1,2,\cdots}$ as follows.
For $2g-2+n\geq 0$,
\beq\label{eqn:TR-des-gn}
\omega_{g,n}(z_1,\cdots,z_n)=
\delta_{(g,n)}^{(0,2)}\frac{d\lambda_1d\lambda_2}{(\lambda_1-\lambda_2)^2}
+\sum_{k_1,\cdots,k_n>0}\<\alpha_{k_1},\cdots,\alpha_{k_n}\>^{\lambda}_{g,n}\frac{d\lambda_1^{-k_1}\cdots d\lambda_n^{-k_n}}{k_1\cdots k_n},
\eeq
where $\lambda_i=\lambda(z_i)$.
For $(g,n)=(0,1)$, $\omega_{0,1}$ has the following form:
\beq\label{eqn:TR-des-01}
\omega_{0,1}(z_1)=\sum_{k\geq 0} v_k\lambda_1^{k-1}d\lambda_1+\sum_{k_1,k_2>0}\<\alpha_{k_1}\>^{\lambda}_{0,2}\frac{d\lambda_1^{-k_1} }{k_1 }.
\eeq
We define the corresponding generating series $Z^{\lambda}_{\cC}(\mathbf{T};\hbar)$ by
\beq\label{def:tau-TR}
Z^{\lambda}_{\cC}(\mathbf{T};\hbar)
=\exp\bigg(\sum_{g\geq 0,n\geq 1}\hbar^{2g-2+n}\sum_{k_1,\cdots,k_n>0}
\<\alpha_{k_1},\cdots,\alpha_{k_n}\>^{\lambda}_{g,n}
\cdot \frac{T_{k_1}\cdots T_{k_n}}{n!}\bigg).
\eeq

\begin{proposition}[\cite{GJZ23}]\label{prop:KP-TR}
	Let $\cC$ be a spectral curve of genus zero possessing a single boundary point
	~\footnote{We remark that \cite{GJZ23} also proves the reduced-KP integrability for higher genus cases by introducing the non-perturbative generating series.}. 
	Let $\lambda$ be a local coordinate near the boundary point such that $x(z)\in \mathbb C[\lambda]$. 
	Then, the generating series $Z^{\lambda}_{\cC}(\mathbf{T};\hbar)$ constitutes a tau-function of the $x$-reduced KP hierarchy with respect to the KP times $\{T_k\}_{k\geq 1}$.
\end{proposition}

\subsection{Topological recursion for generalized Kontsevich model}
In this subsection we show that the multi-differentials $\{\omega^{V}_{g,n}\}$ defined by the GKM with polynomial potential $V(z)$ coincide with the ones $\{\omega_{g,n}\}$ defined by the topological recursion of the spectral curve $\cC=\big(\mathbb P^1, x(z)=V'(z),  y(z)=z\big)$,
proving Theorem \ref{thm:main1}.
Let $x(z)=a_{r}z^{r}+\cdots+a_1 z+a_0$, where $a_r\ne 0$, 
we require that $x(z)$ has only simple critical points, which is satisfied for generic parameters $a_i$.

We first prove a useful lemma for the multi-differentials from topological recursion.
\begin{lemma}\label{lem:regularity-TR}
	Let $\{\omega_{g,n}\}_{g,n\geq 0}$ be the set of multi-differentials defined by the topological recursion on the spectral curve and let $\{z^\gamma\}_{\gamma=1}^{N}$ be the set of critical points.
	For $2g-2+n>0$ and $\gamma=1,\cdots,N$, we have 
	$$
	\omega_{g,n}(z,z_{[n-1]})+\omega_{g,n}(\bar z,z_{[n-1]}) \,  {\text{ is regular near $z=z^{\gamma}$}},
	$$
	and 
	$$
	\frac{\omega_{g,n}(z,z_{[n-1]})}{dx(z)}+\frac{\omega_{g,n}(\bar z,z_{[n-1]})}{dx(z)} \,  {\text{ is regular near $z=z^{\gamma}$}}.
	$$
	Furthermore, for $\gamma=1,\cdots,N$,
	$$
	\frac{\omega_{0,2}(z,\bar z)}{dx(z)}+\frac{\tilde\omega_{0,2}(z,z)}{dx(z)}\, {\text{ is regular near $z=z^{\gamma}$}},
	$$
	where $\tilde \omega_{0,2}(z_1,z_2)=\omega_{0,2}(z_1,z_2)-\frac{dx(z_1)dx(z_2)}{(x(z_1)-x(z_2))^2}$. 
\end{lemma}
\begin{proof}
	For the first part of the Lemma, according to the structure of $\omega_{g,n}$ (equation~\eqref{eqn:TR-CohFT}), we just need to prove that 
	$$
	d\zeta_k^{\bar\beta}(z)+d\zeta_k^{\bar\beta}(\bar z) \,  {\text{ is regular near $z=z^{\gamma}$}}
	$$
	for $k=0,1,2,\cdots$ and $\beta,\gamma=1,\cdots,N$.
	In fact, for each $\gamma=1,\cdots,N$, let $\eta=\eta(z)$ (depends on $\gamma$) be the local airy coordinate near $z^{\gamma}$ defined by $x(z)=x(z^{\gamma})+\frac{1}{2}\eta^2$ satisfying $\lim_{z=z^\gamma}\eta/z=\sqrt{x''(z^\gamma)}$,
	then for $k=0$, $d\zeta_0^{\bar\beta}(z)$ admits the following local expansion:
	$$
	d\zeta_0^{\bar\beta}(z)=-\frac{\delta_{\gamma,\beta}}{\eta^2}d\eta+{\text{regular part}}. 
	$$
	Notice that $\eta(\bar z)=-\eta(z)$, we get
	$$
	d\zeta_0^{\bar\beta}(z)+d\zeta_0^{\bar\beta}(\bar z) \,  {\text{ is regular near $z=z^{\gamma}$}}.
	$$
	Furthermore, since 
	$$
	-d\frac{1}{dx(z)}(\eta^nd\eta)=-(n-1)\eta^{n-2}d\eta,
	$$
	the singular part of $d\zeta_k^{\bar\beta}(z)$ is of the form $f_k^{\bar\beta}(\eta^{-2})d\eta$, where $f_k^{\bar\beta}$ is a polynomial.
	Then by the similar discussion as above, one can see
	$$
	d\zeta_k^{\bar\beta}(z)+d\zeta_k^{\bar\beta}(\bar z) \,  {\text{ is regular near $z=z^{\gamma}$}}.
	$$
	This proves the first part of the Lemma and
	the proof for the second part is similar. 
	
	For the third part of the Lemma, similar as the discussion in the proof of~\cite[Proposition 3.1]{GZ25b}, we first notice the following local expansions:
	$$
	\frac{\omega_{0,2}(z,\bar z)}{dx(z)}=\frac{1}{dx(z)}\frac{d\eta(z)d\eta(\bar z)}{(\eta(z)-\eta(\bar z))^2}
	-c\cdot \frac{d\eta(z)}{\eta(z)}+{\text{regular part}}
	$$
	and
	$$
	\frac{\tilde \omega_{0,2}(z, z)}{dx(z)}=\frac{1}{dx(z)}
	\Big(\frac{d\eta(z)d\eta(z')}{(\eta(z)-\eta(z'))^2}
	-\frac{d\eta(z)^2d\eta(z')^2}{(\eta(z)^2-\eta(z')^2)^2}
	\Big)\Big|_{z'\to z}
	+c\cdot \frac{d\eta(z)}{\eta(z)}+{\text{regular part}}.
	$$
	Then the result follows from the straightforward computations.
	The Lemma is proved.
\end{proof}

\begin{theorem}[=Theorem \ref{thm:main1}]
	\label{thm:GKM-TR}
	Let $V(z)$ be a polynomial in variable $z$ such that $V'(z)$ has only simple critical points, we denote by $\omega^{V}_{g,n}$ the local multi-differentials defined by the GKM with potential $V$, and denote by $\omega_{g,n}$ the multi-differentials defined by the topological recursion on the spectral curve
	$\cC=(\mathbb P^1, x(z)=V'(z),  y(z)=z)$.
	Then for $2g-2+n>0$,
	\beq\label{eqn:GKMw-TRw}
	\omega^{V}_{g,n}(z_1,...,z_n)=\omega_{g,n}(z_1,...,z_n) \, {\text { near }}z_1,\cdots,z_n=\infty.
	\eeq
\end{theorem}
\begin{proof}
	For the spectral curve $\cC$, we take the local coordinate $\lambda=z$ near the boundary point $\infty$, and denote by $Z_{\cC}({\bf T};\hbar)$ the corresponding generating series defined by equation~\eqref{def:tau-TR}.
	Notice that the local expansion of $\omega_{0,1}$ and $\omega_{0,2}$ have the following forms:
	$$\textstyle
	\omega_{0,1}(z_1)=y(z_1)dx(z_1)\in \mathbb C[z_1]\cdot dz_1,\qquad
	\omega_{0,2}(z_1,z_2)=\frac{dz_1dz_2}{(z_1-z_2)^2},
	$$
	we have for $k,l\geq 1$,
	$\<\alpha_k\>_{0,1}^{z}=\<\alpha_k,\alpha_l\>^{z}_{0,2}=0$.
	Therefore, equation~\eqref{eqn:GKMw-TRw} is equivalent to the following identification:
	$$
	Z_{V}({\bf T};\hbar)=Z_{\cC}({\bf T};\hbar).
	$$
	We prove this by showing that both of these two generating series satisfy the $x$-reduced KP integrability, the same string equation, and the same initial condition.
	
	The $x$-reduced KP integrability for $Z_{V}$ has been introduced in~\cite{Ale21a} and the $x$-reduced KP integrability for $Z_{\cC}$ can be derived by Proposition~\ref{prop:KP-TR}.

	Then we  prove that $Z_{\cC}({\bf T};\hbar)$ satisfies the string equation 
	$$
	\frac{\pd}{\pd T_1} Z_{\cC}({\bf T};\hbar)
	=\sum_{m\geq r}c_m L_{-m}Z_{\cC}({\bf T};\hbar),
	$$
	where $c_k$ is defined by $\frac{1}{x'(z)}=\sum_{m=r}^{\infty}c_mz^{-m+1}$. 
	According to Proposition~\ref{prop:GKM-string-wgn}, we just need to prove that the multi-differentials $\omega_{g,n}$ of the spectral curve $\cC$ satisfy  the differential version of the string equation ~\eqref{eqn:GKM-string-wgn},
	i.e.,
	\begin{align}
		&\, \mathop{\Res}_{z_0=\infty} z_0\cdot \omega_{g,n+1}(z_0,z_{[n]})
		=-\sum_{j=1}^n d\Big( \frac{\omega_{g,n}(z_{[n]})}{dx(z_j)} \Big)
		+\delta_{(g,n)}^{(0,2)}\,d_1d_2\frac{x'(z_1)-x'(z_2)}{x'(z_1)x'(z_2)(z_1-z_2)}.
		\label{eqn:GKM-string-wgn-V}
	\end{align}
	We derive this from the topological recursion procedure. 
	Explicitly, we prove the differential version of the string equation for $\omega_{g,n}$ by multiplying $z_0$ and then taking the residue $\mathop{\Res}_{z_0=\infty} $ on the both sides of equation~\eqref{eqn:def-TR}. 
	After doing this, the left-hand side of equation~\eqref{eqn:def-TR} gives $\Res_{z_0=\infty}z_0\cdot \omega_{g,n+1}(z_0,z_{[n]})$.
	For the right-hand side of equation~\eqref{eqn:def-TR}, we first recall that the Bergman kernel $B(z_1,z_2)=\frac{dz_1dz_2}{(z_1-z_2)^2}$ for the Riemann sphere $\mathbb P^1$, thus,
	$$
	K_{\beta}(z_0,z)=\frac{1}{2}\frac{\int_{z'=\bar{z}}^z w_{0,2}(z_0,z')}{(y(z)-y(\bar{z}))dx(z)}
	=\frac{1}{2(z_0-z)(z_0-\bar z)}\frac{dz_0}{dx(z)}.
	$$
	Therefore, for any multi-differential $\omega$, we have 
	\begin{align*}
		\mathop{\Res}_{z_0=\infty}z_0\cdot\sum_{\beta} \mathop{\Res}_{z=z^\beta}K_{\beta}(z_0,z)\cdot \omega(z,\bar z,-)
		=&\, \sum_{\beta}\mathop{\Res}_{z=z^\beta}\Big(\mathop{\Res}_{z_0=\infty}z_0\cdot K_{\beta}(z_0,z)\Big)\cdot \omega(z,\bar z,-)\\
		=&\, -\frac{1}{2}\sum_{\beta}\mathop{\Res}_{z=z^\beta}\frac{1}{dx(z)}\cdot \omega(z,\bar z,-).
	\end{align*}
	Applying this to the right-hand side of equation~\eqref{eqn:def-TR},
	to prove the differential version of the string equation for $\omega_{g,n}$,
	we just need to prove the following equations:
	\beq\label{eqn:TR-string-02}
	-\sum_{\beta}\mathop{\Res}_{z=z^\beta} \frac{1}{dx(z)}\omega_{0,2}(z,z_1)
	\omega_{0,2}(\bar z,z_2) =d_1d_2\frac{x'(z_1)-x'(z_2)}{x'(z_1)x'(z_2)(z_1-z_2)},
	\eeq
	for $2g-2+n>0$,
	\beq\label{eqn:TR-string}
	-\sum_{\beta}\mathop{\Res}_{z=z^\beta} \frac{1}{dx(z)}\omega_{g,n}(z,z_{[n]\setminus \{j\}})
	\omega_{0,2}(\bar z,z_j)
	=-d\Big( \frac{\omega_{g,n}(z_{[n]})}{dx(z_j)} \Big),
	\eeq
	and
	\beq\label{eqn:TR-string-vanishterm}
	\sum_{\beta}\mathop{\Res}_{z=z^\beta} \frac{1}{dx(z)}\omega_{g_1,n_1+1}(z,-)
	\omega_{g_2,n_2+1}(\bar z,-)
	=\sum_{\beta}\mathop{\Res}_{z=z^\beta} \frac{1}{dx(z)}\omega_{0,2}(z,\bar z)
	=0,
	\eeq
	where $2g_i-2+n_i+1>0$ for $i=1,2$. We prove these equations one by one.
	
	For equation~\eqref{eqn:TR-string-02}, we have firstly
	$$
	-\sum_{\beta}\mathop{\Res}_{z=z^\beta} \frac{1}{dx(z)}\omega_{0,2}(z,z_1)\omega_{0,2}(\bar z,z_2)
	=-\sum_{\beta}\mathop{\Res}_{z=z^\beta} \frac{dz d\bar z}{dx(z)}d_1d_2\frac{1}{(z-z_1)(\bar z- z_2)}.
	$$
	Notice that for each $\beta$, $dx(z)$ has a simple root at $z^{\beta}$ and $\lim_{z=z^\beta} \frac{d\bar z}{dz}=-1$, $\lim_{z=z^{\beta}}\bar z=z^\beta$, 
	then we have
	$$
	-\mathop{\Res}_{z=z^\beta} \frac{dz d\bar z}{dx(z)}d_1d_2\frac{1}{(z-z_1)(\bar z- z_2)}
	=d_1d_2\frac{1}{x''(z^\beta)(z^\beta-z_1)(z^\beta- z_2)}.
	$$
	Thus, to prove equation~\eqref{eqn:TR-string-02}, it is sufficient to prove the following equation:
	\beq\label{eqn:B-x}
	\sum_{\beta}\frac{1}{x''(z^\beta)(z_1-z^\beta)(z_2-z^\beta)}
	=\frac{x'(z_1)-x'(z_2)}{x'(z_1)x'(z_2)(z_1-z_2)}.
	\eeq
	Starting from the right hand side of equation~\eqref{eqn:B-x}, the function $\frac{x'(z_1)-x'(z_2)}{x'(z_1)x'(z_2)(z_1-z_2)}$ has poles only at $z_1=z^\beta$ or $z_2=z^\beta$,
	then one has,
	\begin{align}\label{eqn:B-x-cbeta}
		\frac{x'(z_1)-x'(z_2)}{x'(z_1)x'(z_2)(z_1-z_2)}
		=\sum_{\beta,\gamma}\frac{c_{\beta,\gamma}}{(z_1-z^\beta)(z_2-z^\gamma)}.
	\end{align}
	To specify these constants $c_{\beta,\gamma}$,
	we multiply both sides of equation~\eqref{eqn:B-x} by $dz_1$ and then take the residue at $z_1=z^\beta$ to obtain
	$$
	\sum_{\gamma}\frac{c_{\beta,\gamma}}{z_2-z^\gamma}
	=\mathop{\Res}_{z_1=z^\beta}\frac{x'(z_1)-x'(z_2)}{x'(z_1)x'(z_2)(z_1-z_2)}dz_1
	=\frac{1}{x''(z^\beta)\, (z_2-z^\beta)}.
	$$
	Therefore, we have $c_{\beta,\gamma}=\frac{\delta_{\beta,\gamma}}{x''(z^\beta)}$. 
	Then equation~\eqref{eqn:B-x-cbeta} is equivalent to equation~\eqref{eqn:B-x}, which proves equation~\eqref{eqn:TR-string-02}.
	
	For equation~\eqref{eqn:TR-string}, by Lemma~\ref{lem:regularity-TR}, we have known that 
	$\frac{\omega_{g,n}(z,z_{[n]\setminus \{j\}})}{dx(z)}+\frac{\omega_{g,n}(\bar z,z_{[n]\setminus \{j\}})}{dx(z)}$ is regular near $z^{\beta}$, thus we just need to prove the following equation:
	\beq\label{eqn:TR-string-2}
	\sum_{\beta}\mathop{\Res}_{\bar z=z^\beta} \frac{1}{dx(\bar z)}\omega_{g,n}(\bar z,z_{[n]\setminus \{j\}})
	\omega_{0,2}(\bar z,z_j)
	=-d\Big( \frac{\omega_{g,n}(z_{[n]})}{dx(z_j)} \Big),
	\eeq
	where we have used $\Res_{z=z^{\beta}}f(z)dz=\Res_{\bar z=z^\beta}f(z)dz$ and $dx(z)=dx(\bar z)$. 
	In fact, the expression on the left hand side of equation~\eqref{eqn:TR-string-2} has only possible poles at the critical point and $\bar z=z_j$, thus by the Residue theorem, the left hand side of equation~\eqref{eqn:TR-string-2} equals
	$$
	-\mathop{\Res}_{\bar z=z_j} \frac{1}{dx(\bar z)}\omega_{g,n}(\bar z,z_{[n]\setminus \{j\}})
	\omega_{0,2}(\bar z,z_j).
	$$
	Then by using the property of the Bergman kernel: $\Res_{z=z'}f(z)B(z,z')=df(z')$,
	we see the above expression is equal to the right hand side of equation~\eqref{eqn:TR-string-2}.
	This proves equation~\eqref{eqn:TR-string}.
	
	For equation~\eqref{eqn:TR-string-vanishterm}, it can be proved by using the similar discussion as above.
	One just needs to notice that both of the expressions
	$\frac{1}{dx(z)}\omega_{g_1,n_1+1}(z,-)\omega_{g_2,n_2+1}(z,-)$ and
	$$
	\frac{1}{dx(z)}\tilde\omega_{0,2}(z,z)
	=\bigg(\frac{1}{4}\frac{x''(z)^2}{x'(z)^3}-\frac{1}{6}\frac{x'''(z)}{x'(z)^2}\bigg)dz
	$$
	have only possible poles at the critical points.

	Now we prove the initial condition: $Z_{\cC}({\bf 0};\hbar)=1$, which is equivalent to the vanishing of $\omega_{g,0}$ for $g\geq 2$.
	Let $x(z)=\sum_{i=0}^{r}a_iz^i$, where $a_r\ne 0$. 
	By introducing degrees $\deg z:=\frac{1}{r}$ and $\deg a_i=1-\frac{i}{r}$, $i=0,\cdots,r$, one can recursively deduce 
	$$\textstyle
		\deg \omega_{g,n}(z_1,\cdots,z_n)=-\frac{r+1}{r}(2g-2+n).
	$$
	In particular, $\deg \omega_{g,0}=-\frac{r+1}{r}(2g-2)$.
	However, since $\omega_{g,0}$ is power series of $\{\frac{a_i}{a_{r}}\}$ and each $\frac{a_i}{a_{r}}$ is of non-negative degree, one must have $\omega_{g,0}=0$ for $g\geq 2$.
	The Theorem is proved.
\end{proof}

\begin{corollary}
	The locally defined multi-differentials $\omega^{V}_{g,n}$ associated with the GKM with polynomial potential $V$ can be extended to be global multi-differentials on $(\mathbb P^1)^{\times n}$.
\end{corollary}

\subsection{Shifted $r$-spin curve case of the Gukov--Su{\l}kowski conjecture}\label{sec:GS}
Start from a spectral curve $\cC=(\Sigma,x(z),y(z))$, where $x,y$ satisfy the polynomial equation~\footnote{Gukov--Su{\l}kowski also consider the case that $x,y$ satisfy the polynomial equation $A(e^{x},e^{y})=0$, we don't consider this case in the present paper.}
$$
A(x,y)=0.
$$
Gukov--Su{\l}kowski introduced the following Baker--Akhiezer function \cite{GS12} based on the topological recursion,
$$
\tilde{\Psi}(z)=\exp\Big(\sum_{k=0} \hbar^{k-1} S_k(z)\Big),
$$
where $S_0(z):=\int^z y(z) dx(z)$,
$S_1(z):=-\frac{1}{2}\log \frac{dx(z)}{dz}$,
and for $k\geq 2$,
$$
S_k(z):=\sum_{2g-1+n=k}\frac{1}{n!}\int^z\dots\int^z
\omega_{g,n}(z_1,...,z_n).
$$
They proposed that,
there should be a quantum spectral curve $\hat{A}$, defined by quantizing the polynomial $A(x,y)$,
that annihilates the Baker--Akhiezer function:
$$
\hat{A} \, \tilde{\Psi}(z) = 0.
$$
Moreover, they provided a construction of $\hat{A}=\hat{A_0}+\hbar \hat{A_1}+\hbar^2\hat{A_2}+\cdots$ by a recursive procedure and showed that in examples that the quantum curve $\hat{A}$ can be taken as the canonical quantization ($\hat{x}:=x$ and $\hat{y}:=\hbar\frac{d}{dx}$) of the polynomial $A$.

We call curve~\eqref{eqn:spec data} the shifted $r$-spin curve, as it shifts the $r$-spin curve which has $x(z)=z^{r}$.
For the shifted $r$-spin curve, it is clear that the polynomial $A(x,y)$ can be given by
$$
A(x,y)=x(y)-x.
$$
Based on the relationship between topological recursion of curve~\eqref{eqn:spec data} and the GKM, established in Theorem \ref{thm:main1},
we prove that for the shifted $r$-spin curve,
the Gukov--Su{\l}kowski Conjecture holds and the quantum curve is given by the canonical quantization of $A(x,y)$: 
$$
\hat{A}=x(\hat{y})-x.
$$

\begin{corollary} 
	For the shifted $r$-spin curve with parameters $\{a_i\}$:
	$$
	A(x,y)=\sum_{i=0}^{r}a_i y^i -x=0,
	$$
	where $a_r\ne 0$, we assume that $\sum_{i=1}^{r}ia_i y^{i-1}=0$ has only simple roots.
	Let $\tilde {\Psi}(z)$ be the Baker--Akhiezer function associated with the topological recursion on the following spectral curve:
	$$
	\cC=\Big(\mathbb P^1,\quad x(z)=\sum_{i=0}^{r}a_i z^i,\quad y(z)=z\Big).
	$$ 
	Then the following equation holds:
	$$
	\big(x(\hat{y})-\hat{x}\big)\tilde{\Psi}(z)=0,
	$$
	where $\hat{x}:=x(z)\cdot$ and $\hat{y}:=\hbar\frac{d}{dx(z)}$.
\end{corollary}
\begin{proof}
	Firstly, by equation~\eqref{eqn:TR-des-gn}, for $2g-2+n>0$,
	$$
	\int_{\infty}^z\dots\int_{\infty}^z
	\omega_{g,n}(z_1,...,z_n)
	=\sum_{k_1,\cdots,k_n\geq 1}\<\alpha_{k_1},\cdots,\alpha_{k_n}\>_{g,n}^{z}\prod_{i=1}^{n}\frac{z^{-k_i}}{k_i}.
	$$
	Secondly, by the definition of $\omega^{V}_{g,n}$ (equation~\eqref{eqn:def omega^V}) and Theorem~\ref{thm:main1},
	$$
	\<\alpha_{k_1},\cdots,\alpha_{k_n}\>_{g,n}^{z}
	=[\hbar^{2g-2+n}]\frac{\pd^n \log(Z_{V}({\bf T};\hbar))}{\pd T_{k_1}\cdots \pd T_{k_n}}\Big|_{{\bf T=0}}.
	$$
	Then it follows that
	$$
	S_k(z)=[\hbar^{k-1}]\log(Z_{V}({\bf T};\hbar))\big|_{T_n=\frac{z^{-n}}{n},n=1,2,\cdots},
	$$
	and thus
	$$
	\tilde\Psi(z)=e^{\hbar^{-1} S_0(z)+S_1(z)}\Psi_{V}(z),
	$$
	where $\Psi_V(z)$ is the wave function.
	Lastly, comparing equation~\eqref{eqn:KS Q for GKM}, we have
	\begin{align*}
		x(\hat{y})-\hat{x}=
		e^{\hbar^{-1} S_0(z)+S_1(z)} \big(x(Q_V)-x(z)\big) e^{-\hbar^{-1} S_0(z)-S_1(z)}.
	\end{align*}
	The theorem follows immediately from equation~\eqref{eqn:qsc}.
\end{proof}

\begin{remark}
	In \S 3.5 of \cite{Ale21a},
	Alexandrov studied the quantum spectral curve operator for the GKM purely within the KP hierarchy framework, without topological recursion.
	He defined it as an operator that annihilates the wave function,
	thereby providing a quantum counterpart to the spectral curve even in the absence of a classical one (see his Conjecture 2.1).
\end{remark}

\section{Witten's $r$-spin theory and its relation with generalized Kontsevich model}
\label{sec:rspin-GKM}
In this section, we first establish the relationship of the topological recursion on the shifted $r$-spin curve with the shifted total descendent potential $\cD^{r,t}({\bf t};\hbar)$ of the $r$-spin theory.
Then we show the explicit relationship of the multi-differentials $\omega_{g,n}^{V}$ with the shifted descendent invariants of the $r$-spin theory, proving Theorem~\ref{thm:main2}.
As a result, this gives an explicit geometric formulation of the GKM with polynomial potential. 

\subsection{Witten's $r$-spin theory}\label{sec:rspin-geo}
In~\cite{Wit93}, Witten introduced the intersection theory on the moduli space of the $r$-spin curves.
The mathematical definition of the moduli space, denoted by $\Mbar_{g,n}^{r}$, was then constructed by Jarvis~\cite{Jar00}.
The corresponding cohomological class ${\Wit}^{r}_{g,n}(i_1,\cdots,i_n)$, called the Witten top Chern class, was constructed by~\cite{PV01} and then also by Chiodo~\cite{Chi06} and Fan--Jarvis--Ruan~\cite{FJR13}.
Now this theory is also known as the FJRW theory of $A_{r-1}$ singularity.

The intersection number of the $r$-spin theory is defined by
$$
\<\phi_{i_1}\psi^{k_1},\cdots,\phi_{i_n}\psi^{k_n}\>_{g,n}^{r}
:=\int_{\Mbar_{g,n}^{r}}{\Wit}^{r}_{g,n}({i_1},\cdots,{i_n})\psi_1^{k_1}\cdots \psi_n^{k_n}.
$$
Here, $\psi_j$ is the psi-class on $\Mbar_{g,n}^{r}$, $k_j\geq 0$, and the indices $i_j$ take values in $\{0,\cdots,r-2\}$.
The invariant $\<\phi_{i_1}\psi^{k_1},\cdots,\phi_{i_n}\psi^{k_n}\>_{g,n}^{r}$ is vanished unless the following dimension condition holds:
\beq\label{eqn:rspin-dim-condition}
\sum_{i=1}^n\Big(\frac{i_1}{r}+k_i-1\Big)
=\Big(2+\frac{2}{r}\Big)(g-1).
\eeq
The total descendent potential $\cD^{r}({\bf t};\hbar)$ is defined by
\beq\label{def:D-r}
\cD^{r}({\bf t};\hbar):=\exp\bigg(
\sum_{g\geq 0}\hbar^{2g-2}\sum_{n\geq 0}\frac{1}{n!}\<{\bf t}(\psi_1),\cdots,{\bf t}(\psi_n)\>^{r}_{g,n}
\bigg),
\eeq
where ${\bf t}(\psi)=\sum_{k\geq 0}\sum_{i=0}^{r-2}t_k^i\phi_i\psi^k$.

The generalized Witten Conjecture states that via the following coordinate transformation:
\beq\label{eqn:rspin-t-p}
t_n^i({\bf T})=\sqrt{-r}\, (-1)^{n}\, \frac{\Gamma(\frac{i+1}{r}+n+1)}{\Gamma(\frac{i+1}{r})} \, T_{rn+i+1},\qquad n\geq 0,\quad i=0,\cdots,r-2,
\eeq
the total descendent potential $\cD^{r}(\hbar\cdot{\bf t(T)};\hbar)$ gives a tau-function of the $r$KdV hierarchy, with time variables $\{T_k\}$.
Furthermore, $\cD^{r}$ is uniquely determined by the $r$KdV integrability together with the string equation:
$$
\bigg(-\frac{\pd }{\pd t_0^0}+\sum_{k,a}t_{k+1}^a\frac{\pd}{\pd t_k^a}+\frac{1}{2\hbar^2}\sum_{a+b=r-2}t_0^at_0^b\bigg)\cD^{r}({\bf t};\hbar)=0,
$$
and the initial condition $\cD^{r}({\bf 0};\hbar)=1$.
This conjecture was proved by Faber--Shadrin--Zvonkine~\cite{FSZ06} (see also~\cite{FJR13}).

\subsection{Relationship of $r$-spin theory and shifted $r$-spin curve}
To establish the relationship of the $r$-spin theory with the shifted $r$-spin curve, we first introduce the shifted descendent invariants:
$$
\<\phi_{i_1}\psi^{k_1},\cdots,\phi_{i_n}\psi^{k_n}\>_{g,n}^{r,t}
:=\sum_{m\geq 0}\frac{1}{m!}\<\phi_{i_1}\psi^{k_1},\cdots,\phi_{i_n}\psi^{k_n},t,\cdots,t\>_{g,n+m}^{r},
$$
where $t=\sum_{i=0}^{r-2}t^i\phi_i$, and the shifted total descendent potential:
$$
\cD^{r,t}({\bf t};\hbar)
:=\exp\bigg(
\sum_{g\geq 0}\hbar^{2g-2}\sum_{n\geq 0}\frac{1}{n!}\<{\bf t}(\psi_1),\cdots,{\bf t}(\psi_n)\>^{r,t}_{g,n}
\bigg).
$$

Now we consider the shifted $r$-spin curve:
$$
\cC=\big(\mathbb P^1,\quad x(z)=a_{r}z^{r}+\cdots+a_1z+a_0,\quad y(z)=z\big),
$$
where $\{a_i\}$ are parameters such that $a_{r}\ne 0$ and $x(z)$ has only simple critical points.
We denote by $\{\omega_{g,n}\}$ the multi-differentials defined by the topological recursion on $\cC$.
Let $\lambda$ be the local coordinate near $z=\infty$ defined by $x(z)=\lambda^r$ and $\lim_{z\to \infty} \lambda/z=a_r^{1/r}$,
then one gets the corresponding generating series $Z_{\cC}^{\lambda}({\bf T};\hbar)$ by equation~\eqref{def:tau-TR}. 

We have the following theorem.
\begin{theorem}\label{thm:TR-rspin}
	We have the following equality:
	\beq\label{eqn:D=Z}
	Z^{\lambda}_{\cC}\big({\bf T};\hbar/(a_r^{1/r}\sqrt{-r})\big)
	=\cD^{r,t(a)}\big(\hbar\cdot{\bf t(T)};\hbar\big)\big/e^{\frac{1}{\hbar^2}F^{r}_0(t(a))},
	\eeq
	where $F_0^{r}(t)=\sum_{n\geq 3 }\frac{1}{n!}\<t,\cdots,t\>^{r}_{0,n}$, the coordinate transformation ${\bf t}={\bf t(T)}$ is given by equation~\eqref{eqn:rspin-t-p}, and $t^i=t^i(a)$ is defined by
	\beq\label{eqn:ydx-mirrormap}
	y(z)dx(z)=a_{r}^{-1/r}\cdot r\lambda^{r}d\lambda+a_{r}^{-1/r}\cdot\sum_{i=0}^{r-1}t^i(a)\lambda^{i}d\lambda+dO(\lambda^{-1}).
	\eeq
\end{theorem}
\begin{proof}
	We prove the Theorem by showing that the expressions on the both sides of equation~\eqref{eqn:D=Z} are tau-functions of the $r$KdV hierarchy and satisfy the same string equation and initial condition.
	In the follows, we prove these three properties one by one.
	
	Firstly, by~\cite[Theorem I]{GJZ23}, $Z^{\lambda}_{\cC}({\bf T};\epsilon)$, is a tau-function of the $r$KdV hierarchy, with time variables $\{T_k\}$.
	By noticing $\cD^{r,t}({\bf t};\hbar)=\cD^{r}({\bf t}+t;\hbar)$, the $r$KdV integrability for $\cD^{r,t}({\bf t(T)};\hbar)$ follows from the generalized Witten Conjecture.
	
	Secondly, according to~\cite[Theorem 1]{GZ25b}, $Z^{\lambda}_{\cC}({\bf T};\epsilon)$ satisfies the following string equation
	$$
	\bigg(-\frac{1}{a_{r}^{1/r}\epsilon}\frac{\pd }{\pd T_{1}}
	+\frac{1}{r}\sum_{k\geq 1,r\nmid k} (k+r)\, T_{k+r}\frac{\pd}{\pd T_{k}}
	+\frac{1}{2r}\sum_{i+j=r}\Big(iT_i-\frac{t^{i-1}(a)}{a_{r}^{1/r}\epsilon}\Big)
	\Big(jT_j-\frac{t^{j-1}(a)}{a_{r}^{1/r}\epsilon}\Big)\bigg)Z_{\cC}^{\lambda}({\bf T};\epsilon)=0.
	$$
	Comparing this equation with the string equation for $\cD^{r,t}(\hbar\cdot{\bf t};\hbar)$:
	$$
	\bigg(-\frac{1}{\hbar}\frac{\pd }{\pd t_0^0}+\sum_{k,i}t_{k+1}^i\frac{\pd}{\pd t_k^i}+\frac{1}{2}\sum_{i+j=r-2}\Big(t_0^i+\frac{t^i}{\hbar}\Big)\Big(t_0^j+\frac{t^j}{\hbar}\Big)\bigg)\cD^{r,t}(\hbar\cdot {\bf t};\hbar)=0,
	$$
	we see these two equations coincide with each other (up to a sign)
	via $\epsilon=\hbar/(a_r^{1/r}\sqrt{-r})$, ${\bf t}={\bf t(T)}$, and $t^i=t^i(a)$, $i=0,\cdots,r-2$.
	
	Lastly, for the initial condition, since 
	$$
	\cD^{r,t}({\bf 0};\hbar)
	=e^{\sum_{g\geq 0}\hbar^{2g-2}F^{r}_g(t)},
	$$ 
	where $F_g^{r}(t)=\sum_{n\geq 0 }\frac{1}{n!}\<t,\cdots,t\>^{r}_{g,n}$. 
	We just need to prove that $F_g^r(t)=0$ for $g\geq 1$.
	This follows from the fact that 
	$$
	\<\phi_{i_1},\cdots,\phi_{i_n}\>_{g,n}^{r}=0,\qquad
	g\geq 1, 
	$$
	which can be seen from the dimension condition (equation~\eqref{eqn:rspin-dim-condition}).
	The theorem is proved.
\end{proof}
\begin{remark}\label{rem:mirrormap-rspin}
	Setting $a_r=1$ and $a_{r-1}=0$ (so that $x(z)=z^r+a_{r-2}z^{r-2}+\cdots+a_0$), 
	a direct computation yields $t^i(a)=a_i+f_{i}(a_{i+1},\cdots,a_{r-2})$ for some polynomials $f_i$. 
	Conversely, it follows that $a_i(t)=t^i+g_i(t^{i+1},\cdots,t^{r-2})$. 
	This coincides with Saito's flat structure in the deformation theory of $A_{r-1}$-type singularities~\cite{Sai81}, 
	where $\{t^i\}$ serves as a system of flat coordinates on the deformation space. 
	While the relationship of this spectral curve to the total ancestor potential of the shifted Witten class of the $A_{r-1}$-singularity has been established~\cite{DNOPS19}, the present work establishes its relationship to the total descendent potential of the Witten top Chern class.
\end{remark}

\subsection{Relationship of $r$-spin theory and generalized Kontsevich model}\label{sec:Geo-GKM}
The combination of Theorem~\ref{thm:GKM-TR} and Theorem~\ref{thm:TR-rspin} establishes a relationship between the GKM with degree $r+1$ polynomial potential and the shifted $r$-spin theory. 
Note, however, that this connection takes the topological recursion as a bridge and depends on a choice of local coordinate on the spectral curve specified in equation~\eqref{eqn:D=Z}. 
In this subsection, we provide a direct formulation of this relationship via Laplace transforms.
This result constitutes an all-genus analogue of the descendent mirror theorem established by \cite{FLYZ25} for Gromov-Witten theory of semi-projective toric Calabi-Yau 3-orbifolds.

\begin{theorem}
	Consider paths $\tilde\gamma_i=e^{\frac{2i\pi \sqrt{-1}}{r}}[0,\infty)\subset \mathbb C$, and $\gamma_i:=\tilde\gamma_0-\tilde\gamma_i$, we introduce a sequence of classes by
	$$
	\Phi_i(\givz):=\frac{1}{\sqrt{-r}}\sum_{j=1}^{r-1}\big(e^{\frac{2ij\pi\sqrt{-1}}{r}}-1\big)\, \Gamma\Big(\frac{j}{r}\Big)\, \givz^{\frac{j}{r}}\, \phi_{r-1-j},
	$$
	then we have for $2g-2+n>0$,
	\beq\label{eqn:GKM-rspin}
	\int_{\gamma_{i_1}}\!\!\cdots \!\!\int_{\gamma_{i_n}}e^{-\sum_i V'(z_{i})/\givz_i}
	\omega_{g,n}^{V}(z_1,\cdots,z_n)
	=C_r^{2g-2+n}\cdot \left\<\frac{\Phi_{i_1}(\givz_1)}{\givz_1-\psi_1},\cdots,\frac{\Phi_{i_n}(\givz_n)}{\givz_n-\psi_n}\right\>_{g,n}^{r,t(a)}.
	\eeq
	Here $C_r=a_r^{1/r}\sqrt{-r}$ with $a_r$ the coefficient of $z^r$ in $V'(z)$ and $t(a)$ is given by equation~\eqref{eqn:ydx-mirrormap}.
\end{theorem}
\begin{proof}
	We first explain that the integration on the left hand side of equation~\eqref{eqn:GKM-rspin} is well-defined as a formal power series of $\givz^{-\frac{1}{r}}$ (view $\givz\in a_r\cdot\mathbb R_{+}$).
	This is because $\omega^{V}_{g,n}$ can be extended to be a global multi-differential on $(\mathbb P^1)^{\times n}$ according to Theorem~\ref{thm:main1} and $-V'(z)/\givz$ tends to $-\infty$ when $z$ tends to $\infty$ along paths $\gamma_i$.
	In the follows, we denote $V'(z)$ by $x(z)$.
	
	We view $x(z)$ as a $r:1$ covering map from $\mathbb P^1$ to $\mathbb P^1$, then the path $\givz\cdot[0,\infty)$ in the target space has $r$ pre-images.
	We denote these pre-images by $\tilde\gamma'_i$, $i=0,1,\cdots,r-1$. Clearly,
	$$
	\lim_{z\to \infty, \, z\in \tilde\gamma_i}\arg(z)
	=\lim_{z\to \infty, \, z\in \tilde\gamma'_i}\arg(z).
	$$
	Moreover, we define $\gamma'_i=\tilde \gamma'_0-\tilde \gamma'_i$.
	In the follows, we prove equation~\eqref{eqn:GKM-rspin} by two steps. Firstly, 
	\beq\label{eqn:GKM-rspin-1}
	\int_{\gamma'_{i_1}}\!\!\cdots \!\!\int_{\gamma'_{i_n}}e^{-\sum_i V'(z_{i})/\givz_i}
	\omega_{g,n}^{V}(z_1,\cdots,z_n)
	=C_r^{2g-2+n}\cdot \left\<\frac{\Phi_{i_1}(\givz_1)}{\givz_1-\psi_1},\cdots,\frac{\Phi_{i_n}(\givz_n)}{\givz_n-\psi_n}\right\>_{g,n}^{r,t(a)}.
	\eeq
	Secondly, 
	\beq\label{eqn:GKM-rspin-2}
	\int_{\gamma_{i_1}}\!\!\cdots \!\! \int_{\gamma_{i_n}}e^{-\sum_i V'(z_{i})/\givz_i}
	\omega_{g,n}^{V}(z_1,\cdots,z_n)
	=\int_{\gamma'_{i_1}}\!\! \cdots \!\! \int_{\gamma'_{i_n}}e^{-\sum_i V'(z_{i})/\givz_i}
	\omega_{g,n}^{V}(z_1,\cdots,z_n).
	\eeq
	
	For the first step, let $\lambda$ be the local coordinate defined by $x(z)=\lambda^r$ and $\lim_{z\to \infty}\lambda/z=a_r^{1/r}$,
	then $\lambda=\lambda(z)$ can be extended along $\tilde\gamma'_i$.
	Moreover, we have $\lambda(\tilde\gamma'_i)=e^{\frac{2i\pi \sqrt{-1}}{r}} a_r^{1/r}\cdot [0,\infty)$.
	Therefore,
	$$
	\int_{\gamma'_i}e^{-x(z)/\givz}d\lambda^{-k}
	=-\frac{k}{r}\Big(1-e^{-\frac{2ki\pi \sqrt{-1}}{r}}\Big)\int_{0}^{\infty}e^{- x/\givz} x^{-\frac{k}{r}-1}d x.
	$$
	By using $\Gamma(\alpha)=\int_{0}^{\infty} x^{\alpha-1} e^{-x} dx$ and $\Gamma(\alpha+1)=\alpha\, \Gamma(\alpha)$, we get
	$$
	\int_{\gamma'_i}e^{-x(z)/\givz}d\lambda^{-k}
	=\Big(1-e^{\frac{-ki2\pi \sqrt{-1}}{r}}\Big) \Gamma\Big(-\frac{k}{r}+1\Big)\givz^{-\frac{k}{r}}.
	$$
	Introduce local functions $\chi^i_n(\lambda)$ by
	$$
	\chi_n^j=- \frac{(-1)^n}{\sqrt{-r}}\frac{\Gamma(\frac{j+1}{r}+n)}{\Gamma(\frac{j+1}{r})}\lambda^{-rn-j-1},\qquad
	n\geq 0,\quad j=0,\cdots,r-2,
	$$
	then we have
	$$
	\int_{\gamma'_i}e^{-x(z)/\givz}d\chi_{n}^{j}
	=-\frac{(-1)^n}{\sqrt{-r}}\frac{\Gamma(\frac{j+1}{r}+n)}{\Gamma(\frac{j+1}{r})}
	\Big(1-e^{-\frac{2i(j+1)\pi \sqrt{-1}}{r}}\Big) \Gamma\Big(-n-\frac{j+1}{r}+1\Big)\givz^{-n-\frac{j+1}{r}}.
	$$
	Notice that $\Gamma(\frac{j+1}{r}+n)\Gamma\big(-n-\frac{j+1}{r}+1\big) =(-1)^n{\Gamma(\frac{j+1}{r})}{\Gamma\big(\frac{r-1-j}{r}\big)}$, we get that for $j=1,\cdots,r-1$
	$$
	\int_{\gamma'_i}e^{-x(z)/\givz}d\chi_{n}^{r-1-j}
	=\frac{1}{\sqrt{-r}}
	\big(e^{\frac{2ij\pi \sqrt{-1}}{r}}-1\big) \Gamma\Big(\frac{j}{r}\Big)\givz^{-n-1+\frac{j}{r}}.
	$$
	Then by straightforward computations,
	\begin{align*}
		\sum_{j,n}\phi_{r-1-j}\psi^n\int_{\gamma'_i}e^{-x(z)/\givz}d\chi_{n}^{r-1-j}
		=&\, \frac{\Phi_{i}(\givz)}{\givz-\psi}.
	\end{align*}
	Furthermore, by using Theorem~\ref{thm:main1} and Theorem~\ref{thm:TR-rspin}, we have for $2g-2+n>0$,
	$$
	\omega^{V}_{g,n}(z_1,\cdots,z_n)=(a_r^{1/r}\sqrt{-r})^{2g-2+n}\sum_{}\<\phi_{i_1}\psi^{k_1},\cdots,\phi_{i_n}\psi^{k_n}\>_{g,n}^{r,t(a)}
	d\chi_{k_1}^{i_1}\cdots d\chi_{k_n}^{i_n}.
	$$
	Put above computations together, equation~\eqref{eqn:GKM-rspin-1} is proved.
	
	For the second step, we consider loops $C_i:=\gamma_i-\gamma'_i$ and $C_i^{\epsilon}$ which is almost the same as $C_i$ but slightly changes the path near $\infty$ (the difference of $C_i$ and $C_i^{\epsilon}$ lies in a disc $\{z|g(z,\infty)<\epsilon\}$ where $g$ is the metric of the Riemann sphere induced from $\mathbb R^3$) such that $\infty\notin C^{\epsilon}_{i}$.
	Then we have
	\beq\label{eqn:res-wgn}
	\bigg(\int_{\gamma_i}-\int_{\gamma'_i}\bigg)\big(e^{-x(z)/\givz}\omega^{V}_{g,n}(z,-)\big)
	=\lim_{\epsilon\to 0}\oint_{z\in C^{\epsilon}_i}e^{-x(z)/\givz}\omega^{V}_{g,n}(z,-),
	\eeq
	since the expression $e^{-x(z)/\givz}\omega^{V}_{g,n}(z,-)$ is regular at $z=\infty$.
	Clearly, equation~\eqref{eqn:GKM-rspin-2} follows from that the expression on the right hand side of above equation equals zero.
	In fact, given an arbitrary cycle $C\subset \mathbb P^1$ such that $\infty\notin C$,
	$$
	\int_{z\in C}x(z)^md\zeta_{k}^{\bar\beta}(z)
	=-\int_{z\in C}mx(z)^{m-1}\zeta_{k}^{\bar\beta}(z)dx(z)
	=m\int_{z\in C}x(z)^{m-1}d\zeta_{k-1}^{\bar\beta}(z),
	$$
	we get
	$$
	\int_{z\in C}x(z)^md\zeta_{k}^{\bar\beta}(z)
	=m!\int_{z\in C}d\zeta_{k-m}^{\bar\beta}(z).
	$$
	Here $d\zeta^{\bar\beta}_{k-m}(z)=d(-\frac{d}{dx(z)})^{k-m}\zeta_0^{\bar\beta}(z)$, and for $k<m$, we have $d\zeta^{\bar\beta}_{k-m}(z)\in\mathbb C[z]dz$.
	Recall that $k\geq 0$, $d\zeta^{\bar\beta}_{k}(z)$ is meromorphic and has vanishing residue at each pole, we conclude that
	$$\int_{z\in C}x(z)^md\zeta_{k}^{\bar\beta}(z)=0.$$
	Therefore, according to the structure of $\omega_{g,n}$ (equation~\eqref{eqn:TR-CohFT}), equation~\eqref{eqn:res-wgn} holds true and thus equation~\eqref{eqn:GKM-rspin-2} follows immediately.
	The Theorem is proved.
\end{proof}
While standard topological recursion typically requires $dx(z)$ to have only simple roots 
(and thus non-simultaneously vanishing $a_i$, $i=0,\cdots,r-1$), 
both GKM and the shifted $r$-spin theory remain well-defined in the limit $a_i \to 0$ ($i=0,\cdots,r-1$) and $t^j \to 0$ ($j=0,\cdots,r-2$), respectively. 
Therefore, equation~\eqref{eqn:GKM-rspin} remains valid at this limit.

\section{Deformed generalized Kontsevich model and deformed $r$-spin geometry}
\label{sec:deformed-GKM}
In this section, we study the GKM with deformed potentials $V^{\epsilon}_{r}(z)$. We call a deformed potential $V_{r}^{\epsilon}$ {\it admissible} if its second derivative ${V_{r}^{\epsilon}}''(z)$ both has the form
\begin{align}\label{eqn:def Uepsilon}
	{V_{r}^{\epsilon}}''(z)=U_r^{\epsilon}(z)=\frac{U_r(z)}{\prod_{i=1}^{d}(1-\epsilon_i\, z)},
\end{align}
where $U_r(z)$ is an arbitrary polynomial of degree $r-1$, and satisfies the condition
\begin{align}\label{eqn:limit condition}
	\lim_{z\to \infty}{V_{r}^{\epsilon}}''(z)\, dz=\infty.
\end{align}
Clearly, a potential of the form~\eqref{eqn:def Uepsilon} is admissible if and only if $d\leq r$.

Similar to the polynomial cases, we show that the GKM with admissible deformed potential $V_{r}^{\epsilon}$ corresponds to the topological recursion on the spectral curve
$$\cC^{\epsilon}_{r}=\big(\mathbb P^1,\quad x^\epsilon_{r}(z)={V_{r}^{\epsilon}}'(z),\quad y(z)=z\big).$$
Subsequently, we specialize to the case where 
$U_r^{\epsilon}(z)=\frac{r z^{r-1}}{\prod_{i=1}^{d}(1-\epsilon_i\, z)}$, 
and establish its connection with the $r$-spin theory. 
In particular, for $d=1$, this connection confirms a conjecture proposed by Alexandrov~\cite{Ale21a}.

\subsection{Deformed generalized Kontsevich model}
The partition function of the GKM with potential $V_r^{\epsilon}(z)$ can be defined analogously to the procedure in \S \ref{sec:GKMM}. Taking the large $N$ limit ($N\to \infty$) yields the tau-function $Z_{V_r^{\epsilon}}(\mathbf{T};\hbar)$. This model is referred to as the deformed GKM (see~\cite[\S 4]{Ale21a} for further details).

For the deformed GKM, the operators  $x_r^{\epsilon}(z)\cdot$, where $x^\epsilon_{r}(z)={V_{r}^{\epsilon}}'(z)$, and
$$
	Q_{V_r^{\epsilon}}(z)
	=z-\frac{\hbar}{{x_r^{\epsilon}}'(z)} \partial z +\frac{\hbar}{2}\frac{{x_r^{\epsilon}}''(z)}{{x_r^{\epsilon}}'(z)^2}
$$
are still Kac--Schwarz operators of the tau-function $Z_{V_r^{\epsilon}}(\mathbf{T};\hbar)$.
These operators impose specific constraints on the tau-function. 
Precisely, for the operator $x_r^{\epsilon}(z)\cdot$,
the corresponding constraint is the $x$-reduction,
which can be reformulated as the following Hirota quadratic equations (HQEs):
\beq\label{eqn:HQE-Z}
\mathop{\Res}_{z=\infty}x_r^{\epsilon}(z)^k\cdot \Gamma^{+}Z_{V_r^{\epsilon}}({\bf T};\hbar)
\cdot \Gamma^{-}Z_{V_r^{\epsilon}}({\bf T'};\hbar)=0, \qquad k\geq 0,
\eeq
where
$$
\Gamma^{\pm}
= \exp\bigg( \pm \sum_{n\geq 1} z^n T_n \bigg)
\exp\bigg( \mp \sum_{n\geq 1} \frac{1}{n z^n} \frac{\partial}{\partial T_n} \bigg).
$$	
Here, the expression is understood as a Laurent series in $z^{-1}$ for each fixed order in $\epsilon$.
For the operator $Q_{V_r^{\epsilon}}(z)$,
the corresponding constraint takes the form
\begin{align}\label{eqn:hat QCepsilon}
	\widehat{Q_{V_r^{\epsilon}}(z)} Z_{V_r^{\epsilon}}(\mathbf{T};\hbar)=
	f(\epsilon,\hbar)\cdot Z_{V_r^{\epsilon}}(\mathbf{T};\hbar)
\end{align}
for some function $f(\epsilon,\hbar)$ independent of $\mathbf{T}$.
In fact, by Lemma 4.4 of \cite{Ale21a},
these two Kac--Schwarz operators uniquely specify the point in the Sato-Grassmannian corresponding to this tau-function.

We introduce the differential version of connected $n$-point function of this model similar to what we have done in \S \ref{sec:diff version}.
That is,
for any $(g,n)$,
\begin{align*}
	\omega^{V_r^{\epsilon}}_{g,n}(z_1,...,z_n):=
	\sum_{i_1,...,i_n=1}^\infty
	\frac{1}{i_1\cdots i_n}
	[\hbar^{2g-2+n}] 
	\left.\frac{\partial^n \log Z_{V_r^{\epsilon}}(\mathbf{T};\hbar)}
	{\partial T_{i_1} \dots \partial T_{i_n}}
	\right|_{\mathbf{T}=0}
	d z_1^{-i_1} \cdots d z_n^{-i_n}.
\end{align*}
In the following lemma,
we establish the differential version of the constraint provided by the Kac--Schwarz operator $Q_{V_r^{\epsilon}}(z)$.
\begin{lemma}\label{Lem:string-deformed-GKM}
	We assume that $\frac{1}{U_r^{\epsilon}(z)}=\sum_{k=r-d}^\infty c_k z^{-k+1}$.
	Then the $W_{1+\infty}$ realization of the operator $Q_{V_r^{\epsilon}}(z)$ is given by
	\begin{align*}
		\widehat{Q_{V_r^{\epsilon}}(z)}
		=-\frac{\partial}{\partial T_1}
		+\sum_{k=r-d}^\infty c_k L_{-k}.
	\end{align*}
	 As a consequence,
	 the constraint \eqref{eqn:hat QCepsilon} is equivalent to
	 \begin{align*}
	 	\mathop{\Res}_{z_0=\infty} z_0\cdot \omega^{V_r^{\epsilon}}_{g,n+1}(z_0,z_{[n]})
	 	=&-\sum_{j=1}^n d\Big( \frac{\omega^{V_r^{\epsilon}}_{g,n}(z_{[n]})}{dx_r^{\epsilon}(z_j)} \Big)
	 	+\delta_{(g,n)}^{(0,2)}\,d_1d_2\frac{{x_r^{\epsilon}}'(z_1)-{x_r^{\epsilon}}'(z_2)}{{x_r^{\epsilon}}'(z_1){x_r^{\epsilon}}'(z_2)(z_1-z_2)}
	 	&-\delta_{n,0}f_g,
	 \end{align*}
	 where $f_g$ is the coefficient of $\hbar^{2g-1}$ in $f(\epsilon,\hbar)$.
\end{lemma}
\begin{proof}
	The proof is completely the same as that in Proposition \ref{prop:GKM-string-wgn}.
\end{proof}

\subsection{Topological recursion for deformed generalized Kontsevich model}
Now we consider the spectral curve
$$
\mathcal{C}_{r}^{\epsilon} = \bigl(\mathbb{P}^1,\quad x^{\epsilon}_r(z),\quad y(z) = z \bigr).
$$
By applying topological recursion procedure, we obtain the differentials $\{\omega^{\epsilon}_{g,n}\}$. 
Moreover, choosing the local coordinate $\lambda = z$ near $z=\infty$, we define the invariants $\langle - \rangle^{\epsilon, z}_{g,n}$ and the generating series $Z_{\cC_{r}^{\epsilon}}(\mathbf{T};\hbar)$; see §\ref{sec:TR-CohFT-KP}.
We emphasize that the differentials $\{\omega^{\epsilon}_{g,n}\}$, the invariants $\langle - \rangle^{\epsilon, z}_{g,n}$, and the generating series $Z_{\cC_{r}^{\epsilon}}(\mathbf{T};\hbar)$ are all regarded as $\epsilon$-adic objects.

\begin{proposition}
	\label{prop:x-red-KP-deformed-GKM}
	The generating series $Z_{\cC_{r}^{\epsilon}}(\mathbf{T};\hbar)$ is a KP tau-function, with KP times $\{T_k\}$, and satisfies the $x$-reduction.
\end{proposition}
\begin{proof}
	The proposition is equivalent to the validity of the following HQEs:
	$$
	\mathop{\Res}_{z=\infty}x_{r}^{\epsilon}(z)^k\cdot \Gamma^{+}Z_{\cC_{r}^{\epsilon}}({\bf T};\hbar)
	\cdot \Gamma^{-}Z_{\cC_{r}^{\epsilon}}({\bf T'};\hbar)=0,\qquad k\geq 0.
	$$
	Following the argument in~\cite[\S 5]{GJZ23}, one can prove these equations as follows:
	one first transforms these HQEs into the HQEs for the total ancestor potential, then into the HQEs for the Witten–Kontsevich tau-function, and finally proves them using the Witten–Kontsevich theorem.
	Since the proof follows the same steps as in~\cite{GJZ23} (for the case where $x$ is meromorphic and behaves like a polynomial near its unique boundary point), we only explain why their argument applies in our setting. 
	First, $x^{\epsilon}_{r}$ can be regarded as a meromorphic function in the $\epsilon$-adic sense.
	Second, the elements $\frac{1}{dx}$ and $d\zeta^{\bar\beta}_k(z)$ do not introduce additional poles beyond the critical points and the boundary point, again in the $\epsilon$-adic sense.
	Third, the critical points $z^\beta$ are independent of the deformation parameter $\epsilon$.
	These observations ensure that the construction of~\cite{GJZ23} carries over to the present situation without further modification, and repeating their argument yields the desired HQEs.
\end{proof}
\begin{theorem}
\label{thm:GKM-TR-deformed}
	For $g\geq 0$, $n\geq 1$, and $2g-2+n>0$, we have 
	$$
	\omega^{V_{r}^{\epsilon}}_{g,n}(z_1,\cdots,z_n)=\omega^{\epsilon}_{g,n}(z_1,\cdots,z_n)\, 
	{\text{ near }}\,  z_1,\cdots,z_n=\infty.
	$$
\end{theorem}
\begin{proof}
	Similar as the proof of Theorem~\ref{thm:GKM-TR}, we just need to check that the generating series $Z_{V^{\epsilon}_{r}}$ and $Z_{\cC_{r}^{\epsilon}}$ satisfy the same $x$-reduced KP integrability and the same string equation (they do have different initial condition: $Z_{V^{\epsilon}_{r}}({\bf 0};\hbar)=1$ and $Z_{\cC_{r}^{\epsilon}}({\bf 0};\hbar)=\exp(\sum_{g\geq 2}\hbar^{2g-2}\omega_{g,0}^{\epsilon})$).
	First, the $x$-reduced KP integrability has been established on both sides, see equation~\eqref{eqn:HQE-Z} and Proposition~\ref{prop:x-red-KP-deformed-GKM}.
	Second, the TR side of the differential version of the string equation can be proved by using the same method as that in the proof of Theorem~\ref{thm:GKM-TR}, the only difference is that $\mathop{\Res}_{z=\infty}\frac{\tilde \omega_{0,2}(z,z)}{dx^{\epsilon}_r(z)}$ is not vanished. 
	For example, for $U_{r}^{\epsilon}(z)=U_r(z)/(1-\epsilon z)$,
	$$
	\mathop{\Res}_{z=\infty}\frac{\tilde \omega_{0,2}(z,z)}{dx^{\epsilon}_r(z)}
	=\mathop{\Res}_{z=\infty}\bigg(\frac{1}{4}\frac{{x^{\epsilon}_r}''(z)^2}{{x^{\epsilon}_r}'(z)^3}
	-\frac{1}{6}\frac{{x^{\epsilon}_r}'''(z)}{{x^{\epsilon}_r}'(z)^2}\bigg)dz
	=-\frac{\epsilon^2}{12}\mathop{\Res}_{z=\infty}\frac{dz}{(1-\epsilon z)U_r(z)}.
	$$
	We note here that the residue is in $\epsilon$-adic, i.e, we should understand $\frac{1}{1-\epsilon z}$ as $\sum_{n\geq 0}\epsilon^n z^n$, thus the result is 
	$\frac{\epsilon}{12U_r(1/\epsilon)}$.
	This in fact determines $f(\epsilon,\hbar)=\frac{\hbar\cdot\epsilon}{24U_r(1/\epsilon)}$ in equation~\eqref{eqn:hat QCepsilon}.
	The proof is finished.
\end{proof}

\subsection{Relationship of deformed generalized Kontsevich model and $r$-spin geometry}
\label{sec:deformed-GKM-geo}
Now we focus on the case that $U_r^{\epsilon}= \frac{rz^{r-1}}{\prod_{i=1}^{d}(1-\epsilon_i z)}$, 
and establish the relationship between the deformed GKM and $r$-spin theory.
For the general case in which $U_r(z)$ is an arbitrary polynomial of degree $r-1$, an analogous correspondence holds, but is accompanied by a non-trivial mirror map $t = t(a)$, as in the polynomial $U(z)$ function cases.

Define $\lambda=\lambda(z)$ by $\lambda^r=x_{r}^{\epsilon}(z)$ and $\lim_{z\to 0}\lambda/z=1$, then we have 
$\lambda=z+O(z^2)$ and, conversely, $z=\lambda+O(\lambda^2)$.
Following similar discussions in~\cite{DOSS14,Eyn14}, 
we consider the local expansion of the Bergman kernel:
$$
\frac{dz_1dz_2}{(z_1-z_2)^2}=\frac{d\lambda_1 d\lambda_2}{(\lambda_1-\lambda_2)^2}
+\sum_{m,n\geq 0}Q_{m,n}\lambda_1^m\lambda_2^n d\lambda_1d\lambda_2,
$$
where $z_i$ are understood as functions of $\lambda_i$, i.e., $z_i=z_i(\lambda_i)$ for $i=1,2$.
Let $\{\phi_i\}$ be a basis for the state space of the $r$-spin theory, and let $\{\phi^i\}$ be its dual basis with respect to the pairing $\eta$. We define the $V$-matrix $\mathscr V^{r,\epsilon}(z,w) = \sum_{k,l\geq 0} \mathscr V_{k,l} z^k w^l$. Its matrix elements $\mathscr V_{k,l}^{a,b} = \eta(\phi^a, \mathscr V_{k,l} \phi^b)$ are given by:
$$
\mathscr V^{a,b}_{k,l}=\frac{(-1)^{k+l}}{(-r)}\cdot 
\prod_{i=0}^{k-1}\Big(\frac{a+1}{r}+i\Big)\cdot \prod_{j=0}^{l-1}\Big(\frac{b+1}{r}+j\Big)
\cdot Q_{rk+a,rl+b}.
$$
Next, we define the $R$-matrix $\mathscr R^{r,\epsilon}(z)$ via the relation:
$$
\mathscr V^{r,\epsilon}(z,w)=\frac{{\rm I} - \mathscr R^{r,\epsilon,*}(-z)\mathscr R^{r,\epsilon}(-w)}{z+w},
$$
where $\mathscr R^{r,\epsilon,*}$ is the adjoint of $\mathscr R^{r,\epsilon}$ with respect to $\eta$.
The proof that the $R$-matrix is well-defined follows from arguments similar to those presented in~\cite[Appendix B]{Eyn14}.
Similarly, we consider the local expansion of $y(z)$:
$$
y(z)=z=\sum_{n\geq 0}h_n \lambda^n.
$$
We define the $T$-vector $\mathscr T^{r,\epsilon}(z)=\sum_{k\geq 1} \mathscr T^a_k \phi_a z^k$.
Its vector elements $\mathscr T^a_k $ are given by
$$
\mathscr T^{a}_k(z)=\delta_{k,1}\delta_{a,0}-r\cdot (-1)^{k}\cdot \prod_{i=0}^{k}\Big(\frac{a+1}{r}+i\Big)\cdot h_{rk+a+1}.
$$

Recall the Witten top Chern class ${\Wit}^{r}_{g,n}$, we define the corresponding cohomological field theory (CohFT) $\Omega^{r}_{g,n}$ by 
$$
\Omega^{r}_{g,n}(\phi_{i_1},\cdots,\phi_{i_n}):=\pi_{*}{\Wit}^{r}_{g,n}(i_1,\cdots,i_n),
$$
where $\pi:\Mbar^{r}_{g,n}\to \Mbar_{g,n}$ is the forgetful map defined by forgetting the $r$-spin structure.
We introduce the deformed CohFT $\Omega^{r,\epsilon}$ by the $R$- and $T$-actions on $\Omega^{r}$:
$$
\Omega^{r,\epsilon}=\rR^{r,\epsilon}\cdot \mathscr T^{r,\epsilon}\cdot \Omega^{r}.
$$
We refer readers to~\cite{PPZ15} for the definition of the actions of $R$-matrix and $T$-vector on a CohFT.
We introduce the generating series
$$
\mathscr D^{r,\epsilon}({\bf t};\hbar):=
\exp\bigg(\sum_{2g-2+n>0}\frac{\hbar^{2g-2}}{n!}
\int_{\Mbar_{g,n}}\Omega^{r,\epsilon}\big({\bf t}(\psi_1),\cdots,{\bf t}(\psi_n)\big)\bigg),
$$
where ${\bf t}(\psi)=\sum_{k\geq 0}\sum_{i=0}^{r-2}t_k^i\phi_i\psi^k$.

\begin{theorem}
	We have the following equality
	$$
	Z_{V_r^{\epsilon}}\big({\bf T};\hbar/\sqrt{-r}\, \big)
	=\mathscr D^{r,\epsilon}(\hbar\cdot {\bf t(T)};\hbar)/\mathscr D^{r,\epsilon}({\bf 0};\hbar).
	$$
	Here the coordinate transformation ${\bf t(T)}$ is given as follows: introduce the polynomials
	$$
	\varphi_n^{b}(z)=-\frac{(-1)^n}{\sqrt{-r}}\prod_{i=0}^{n-1}\Big(\frac{b+1}{r}+i\Big)\,\cdot 
	\Big[\frac{1}{\lambda(z)^{rn+b+1}}\Big]_{-}\in \mathbb C[\epsilon][z^{-1}],
	$$
	and define the quantization: $\widehat{z^{-k}}:=kT_k$,
	then for $m\geq 0$ and $a=0,\cdots,r-2$,
	$$
	t_m^a({\bf T})=\sum_{n=0}^{m}\sum_{b=0}^{r-2}(\rR_{m-n})^{a}_b \,\cdot \widehat{\varphi_n^{b}(z)}.
	$$
\end{theorem}
\begin{proof}
	The proof proceeds by verifying that both sides of the equality satisfy the characterizing properties (the $x$-reduced KP integrability, the string equation, and the initial condition) of the generating function.
	
	We note firstly that as Givental's reconstruction theorem, we have
		$$
		\mathscr D^{r,\epsilon}({\bf t};\hbar)=[\widehat{\rR^{r,\epsilon}}\cD^{r}]({\bf t};\hbar)
		=[e^{\frac{\hbar^{2}}{2}\rV^{r,\epsilon}(\pd_t,\pd_t)}\cD^{r}]({\rR^{r,\epsilon}}^{-1}\tilde{\bf t};\hbar).
		$$
	We refer readers to~\cite{Giv01} for the explanation of this formula, see also~\cite[\S 1]{GJZ23}.
	Then, analogous to the approach in~\cite[\S 5]{GJZ23}, it can be shown that the expression $[\widehat{\rR^{r,\epsilon}}\cD^{r}](\hbar\cdot {\bf t(T)};\hbar)$ constitutes a tau-function of the $x$-reduced KP hierarchy. While the proof in~\cite{GJZ23} relies on the Witten--Kontsevich theorem, our case necessitates the application of the generalized Witten Conjecture for $r$-spin theory, which has been established in~\cite{FJR13,FSZ06}.
	
	Furthermore, it is straightforward to verify that the string equation for $Z_{V_{r}^{\epsilon}}$ is equivalent to the corresponding string equation in $r$-spin theory under the prescribed coordinate transformation and $R$-matrix action.
	
	Finally, the equality is confirmed by the consistency of the initial conditions, which are trivial in this setting. 
	
	Given these three conditions (the reduced KP integrability, the string equation, and the initial data), the identity follows from the uniqueness of the tau function. We omit the computational details here.
\end{proof}

\begin{remark}
	In~\cite{Ale21a}, Alexandrov proposed a conjecture stating that the deformed GKM with potential $V_r^{\epsilon}= \int^{z}\int^{z}\frac{z^{r-1}}{1-\epsilon z}dzdz$, after a shift of variables,
	describes interesting enumerative geometry invariants in the $r$-spin case.
	Our results confirm Alexandrov’s conjecture at the level of CohFTs.
	From a geometric perspective, the $R$-matrix $\rR^{r,\epsilon}$ plays the role of the operator $\hat{\Delta}$ in~\cite[Theorem 1]{CG07} within twisted Gromov–Witten theory. 
	We expect such a geometric construction to exist in $r$-spin theory that yields the $R$-matrix $\rR^{r,\epsilon}$.
\end{remark}


\end{document}